\documentclass[a4paper]{llncs}
\spnewtheorem{thm}{Theorem}[section]{\bfseries}{\itshape}

\usepackage{amssymb}
\usepackage{amsmath}
\usepackage{booktabs}
\usepackage{color}
\usepackage{complexity}
\usepackage{float}
\usepackage{listings}
\usepackage{times}
\usepackage{hyperref}
\usepackage[dvipsnames]{xcolor}
\usepackage{stmaryrd}
\usepackage{xspace}
\usepackage[all]{xy}
\usepackage{subcaption}
\usepackage{tikz}
\usepackage{url}
\usepackage{paralist}
\usepackage{soul}

\colorlet{shadecolor}{gray!20}

\usetikzlibrary{decorations.pathreplacing,calc,positioning}

\DeclareMathAlphabet{\mathsl}{OT1}{cmr}{m}{sl}

\long\def\omitthis#1{\relax}
\newcommand\longversion[1]{#1}

\newcommand{\red}[1]{\textcolor{Mahogany}{#1}}
\newcommand{\blue}[1]{\textcolor{NavyBlue}{#1}}

\long\def\comment#1{}
\newcommand{\eg}{{\em e.g.}}
\newcommand{\ie}{{\em i.e.}}

\newcommand{\Ascr}{\mathcal{A}}

\newcommand{\Cscr}{\mathcal{C}}

\newcommand{\Sscr}{\mathcal{S}}
\newcommand{\Rscr}{\mathcal{R}}
\newcommand{\Tscr}{\mathcal{T}}
\newcommand{\Wscr}{\mathcal{W}}

\newcommand{\Escr}{\mathcal{E}}

\newcommand{\CS}{\mathcal{CS}}

\newcommand{\quotes}[1]{``#1''}

\long\def\comment#1{\relax}

\newcommand\GS{\mathcal{GS}}

\newcommand\At{\ensuremath{\mathsf{At}}\xspace}

\newcommand\Time{\ensuremath{\mathsf{Time}}\xspace}

\newcommand\Tick{\ensuremath{\mathsf{Tick}}\xspace}

\newcommand\Event{\ensuremath{\underline{\mathsf{Event}}}\xspace}
\newcommand\Attended{\ensuremath{\underline{\mathsf{Attended}}}\xspace}

\newcommand\Flight{\ensuremath{\mathsf{Flight}}\xspace}

\newcommand\airport{\ensuremath{\mathsf{airport}}\xspace}

\newcommand\Center{\ensuremath{\mathsf{center}}\xspace}
\newcommand\main{\ensuremath{\mathsf{main}}\xspace}

\newcommand\no{\ensuremath{\mathsf{no}}\xspace}
\newcommand\id{\ensuremath{\mathsf{id}}\xspace}

\newcommand\done{\ensuremath{\mathsf{done}}\xspace}

\newcommand\FRA{\ensuremath{\mathsf{FRA}}\xspace}

\newcommand\DBV{\ensuremath{\mathsf{DBV}}\xspace}

\newcommand{\gap}{\hspace*{1pt} \\}

\newcommand{\todo}[1]{\blue{\textbf{To-do}: #1}}
\newcommand{\jesse}[1]{\textcolor{red}{\textbf{JC:} #1}}
\newcommand{\tajana}[1]{\textcolor{red}{\textbf{TBK:} #1}}
\newcommand{\varmax}{\eta}
\newcommand{\smalltt}[1]{{\small\texttt{#1}}}
\newcommand{\var}{\textrm{Var}}
\newcommand{\cons}{\textrm{Cons}}
\newcommand{\Gcons}{\mathcal{G}}
\newcommand{\GTerms}{\mathcal{G}_{\textrm{Terms}}}
\newcommand{\values}{\textrm{Values}}

\newcommand{\fresh}{\textrm{Fresh}}
\newcommand{\FOvar}{\mathcal{V}_{\textrm{FO}}}
\newcommand{\Tvar}{\mathcal{V}_{\Time}}
\newcommand{\dom}{\textrm{dom}}

\newboolean{comments}
\setboolean{comments}{true}
\ifthenelse{\boolean{comments}}{}{
\renewcommand{\jesse}[1]{}
\renewcommand{\tajana}[1]{}
\renewcommand{\todo}[1]{}
}

\newcommand\SUR{\ensuremath{\mathsf{SUR}}\xspace}

\newcommand\GUR{\ensuremath{\mathsf{GUR}}\xspace}

\newcommand{\tup}[1]{\langle#1\rangle}

\newcommand\lra{\longrightarrow}

\title{Technical Report: Time-Bounded Resilience}

\author{
Tajana {Ban Kirigin}\inst{1}
\and
Jesse Comer\inst{2}
\and
Max Kanovich\inst{3}
\and
Andre Scedrov\inst{2}
\and 
Carolyn Talcott\inst{4}
}

\institute{
Faculty of Mathematics, University of Rijeka, HR \email{bank@math.uniri.hr}
\and
 University of Pennsylvania, Philadelphia, USA
\email{scedrov@math.upenn.edu,jacomer@seas.upenn.edu}
 \and
University College, London, UK 
\email{m.kanovich@ucl.ac.uk}
\and
  SRI International, USA
    \email{carolyn.talcott@sri.com}
}

\newboolean{longversion}
\setboolean{longversion}{true}

\begin{document}
\maketitle
\begin{abstract}
Most research on system design has focused on optimizing \emph{efficiency}. However, insufficient attention has been given to the design of systems optimizing \emph{resilience}, the ability of systems to adapt to unexpected changes or adversarial disruptions. In our prior work, we formalized the intuitive notion of resilience as a property of cyber-physical systems by using a multiset rewriting language with explicit time. In the present paper, we study the computational complexity of a formalization of \emph{time-bounded} resilience problems for the class of \emph{$\varmax$-simple} progressing planning scenarios, where, intuitively, it is simple to check that a system configuration is critical, and only a bounded number of rules can be applied in a single time step. We show that, in the time-bounded model with $n$ (adversarially-chosen) disruptions, the corresponding time-bounded resilience problem for this class of systems is complete for the $\Sigma^\P_{2n+1}$ class of the polynomial hierarchy, \PH. To support the formal models and complexity results, we perform automated experiments for time-bounded verification using the rewriting logic tool Maude.
\end{abstract}

\section{Introduction}
\label{sec:intro}
Resilience is \quotes{the ability of a system to adapt and respond to change (both environmental and internal)}~\cite{bloomfield2020towards}. In recent years, the task of formally defining and analyzing this intuitive notion has drawn interest across domains in computer science, ranging from systems engineering~\cite{madni2009towards,neches2013towards}, particularly cyber-physical systems (CPS)~\cite{bennaceur2019modelling,koutsoukos2018sure}, to artificial intelligence~\cite{olowononi2021resilient,eigner2021towards,prasad2022towards,goel2024adversarial}, programming languages~\cite{cunningham2014resilient,hukerikar2012programming}, algorithm design~\cite{ferarro2010experimental,caminiti2017resilient}, and more. Our previous work in \cite{Alturki2022resilience} was particularly inspired by Vardi's paper \cite{vardi2020efficiency}, in which he articulated a need for computer scientists to reckon with the trade-off between efficiency and resilience.

In \cite{Alturki2022resilience}, we formalized resilience as a property of timed multiset rewriting (MSR) systems~\cite{kanovich17mscs,kanovich2021Guttman}, which are suitable for the specification and verification of various goal-oriented systems. Although the related verification problems are undecidable in general, it was shown that these problems are \PSPACE-complete for the class of \emph{balanced} systems, in which facts are of bounded size, and rewrite rules do not change configuration size. A primary challenge in~\cite{Alturki2022resilience} was the formalization of the disruptions against which systems must be resilient. This was achieved by separating the system from the environment, delineating between rules applied by the system and those imposed on the system, such as changes in conditions, regulations, or mission objectives.

\subsubsection{Main Contributions.}
This paper formalizes the notion of \emph{time-bounded} resilience. We focus on the class of $\varmax$-simple progressing planning scenarios (PPS) and investigate the computational complexity of the corresponding verification problem. Time-bounded resilience is motivated by bounded model checking and automated experiments, which can help system designers verify properties and find counterexamples where their specifications do not satisfy time-bounded resilience. Moreover, bounded versions of resilience problems arise naturally when the missions of the systems being modeled are necessarily bounded at some level. The main contributions of the paper are as follows.
\begin{enumerate}
\item We define time-bounded resilience as a property of planning scenarios. Intuitively, a resilient system can accomplish its mission within the given time bounds, even in the presence of a bounded number of disruptions (cf. Definition \ref{d-trace-ab-recursive}).
\item We investigate the computational complexity of time-bounded resilience problems, showing that for the class of $\varmax$-simple PPSs with facts of bounded size~\cite{kanovich17jcs}, the time-bounded resilience problem with $n$ updates is complete for the $\Sigma^\P_{2n+1}$ class of the polynomial hierarchy, $\PH$ (Corollary \ref{thm:sigma2n+1-complete}).
\item We demonstrate that our formalization can be automated, using the rewriting logic tool Maude to perform experiments verifying time-bounded resilience (Section \ref{sec:implementation}).
\end{enumerate}

\subsubsection{Expository Example.} In \cite{Alturki2022resilience}, our study of resilience was motivated by current research into CPSs that perform complex, safety-critical tasks in hostile and unpredictable environments, often autonomously. In this paper, we expand our perspective to consider resilience properties of a broad class of multi-agent systems. For expository purposes, we utilize a running example of a researcher planning travel to attend and present research at a conference. The system rules represent actions of the researcher, while update rules represent travel disruptions and changes to the conference organization. Ultimately, the travel planning process is pointless if the researcher does not arrive at his destination in time for the main event. Consequently, the researcher desires to establish a \emph{resilient} plan, which will allow him to accomplish his goal in spite of some bounded number of disruptions. Details of this planning scenario will be developed throughout Section~\ref{sec:msr}, and our Maude implementation in Section~\ref{sec:implementation} will be used to analyze its resilience.

\subsubsection{Outline.} Section~\ref{sec:msr} reviews the timed MSR framework used in Section~\ref{sec:bounded} to define time-bounded resilience. In Section~\ref{sec:complexity-b-time}, we investigate the complexity of the verification problem. Section~\ref{sec:implementation} showcases our results on automated verification obtained using Maude. In Section~\ref{sec:future}, we conclude with a discussion of related and future work.

\section{Multiset Rewriting Systems}
\label{sec:msr}
In this section, we review the framework of timed MSR models introduced in our previous work~\cite{kanovich16formats,kanovich17jcs,kanovich17mscs}.

\subsection{The Rewriting Framework}

\subsubsection*{Terms and Formulas.} We fix a finite first-order alphabet $\Sigma$ with constant, function, and predicate symbols, together with a finite set $\mathcal{B}$ of \emph{base types}. Each constant is associated with a unique base type, and we write $\Sigma_{\cons}$ to denote the set of all constants in $\Sigma$. Each predicate symbol $R$ (resp. function symbol $f$) is associated with a unique \emph{tuple type} (resp. \emph{arrow type}) $b_1 \times \hdots \times b_k$ (resp. $b_1 \times \hdots \times b_k \to b$), where $b_1,\hdots,b_k,b \in \mathcal{B}$ and $k$ is the arity of $R$ (resp. $f$). We also assume that $\Sigma$ contains a special predicate symbol $\Time$ with arity zero (more on this later).

We fix sets $\FOvar$ of (first-order) \emph{variables} and $\Gcons$ of \emph{ground constants}, disjoint from each other and from $\Sigma$, where each element in $\FOvar \cup \Gcons$ has an associated base type in $\mathcal{B}$. We further assume that $\FOvar$ and $\Gcons$ each contain countably infinitely-many elements associated to each base type. These ground constants will be used to instantiate variables ``created'' by rewrite rules. \emph{Terms} over $\Sigma$ are constructed according to the grammar
$$t := x \mid c \mid f(t_1,\hdots,t_k),$$
where $x$ is in $\FOvar$, $c$ is in $\Sigma_{\cons}$, $f$ is a function symbol of type $b_1 \times \hdots \times b_k \to b$, and each $t_i$ is a term of type $b_i$ for $i \leq k$ (in which case $f(t_1,\hdots,t_k)$ is a term of type $b$). \emph{Ground terms} over $\Sigma$ are constructed similarly:
$$a := d \mid c \mid f(a_1,\hdots,a_k),$$
where $d$ is in $\Gcons$, $c$ is in $\Sigma_{\cons}$, $f$ is a function symbol of type $b_1 \times \hdots \times b_k \to b$, and each $a_i$ is a ground term of type $b_i$ for $i \leq k$ (in which case $f(a_1,\hdots,a_k)$ is a ground term of type $b$). We write $\GTerms$ for the collection of ground terms over $\Sigma$. If $R$ is a predicate symbol of type $b_1 \times \hdots \times b_k$ and $t_1,\hdots,t_k$ are terms of type $b_1,\hdots,b_k$, respectively, then $R(t_1,\hdots,t_k)$ is an \emph{atomic formula}. Similarly, if $a_1,\hdots,a_k$ are ground terms of type $b_1,\hdots,b_k$, respectively, then $R(a_1,\hdots,a_k)$ is an \emph{atomic fact} (or just \emph{fact}).

\subsubsection*{Modeling Discrete Time.} We fix a countably infinite set $\Tvar = \{ T_i \mid i \in \mathbb{N} \}$ of \emph{time variables}. \emph{Timestamped atomic formulas} are of the form $F@(T+D)$, where $F$ is an atomic formula, $T$ is a time variable, and $D$ is a natural number; note that if $D = 0$, we prefer to write $F@T$ instead of $F@(T+0)$. \emph{Timestamped facts} are of the form $F@t$, where $F$ is an atomic fact and $t \in \mathbb{N}$ is its \emph{timestamp}. For brevity, we frequently refer to timestamped facts simply as facts. Clearly, given a timestamped atomic formula $F@(T+D)$, we can obtain a timestamped fact $G@t$ by uniformly substituting ground terms for variables in $F$ and setting $t = N+D$ for some natural number $N$.

\subsubsection*{Configurations and Rewrite Rules.}
\emph{Configurations} are multisets of timestamped facts $\Sscr = \{\Time@t, F_1@t_1, \ldots, F_n@t_n\}$ with exactly one occurrence of a $\Time$ fact whose timestamp $t$ is the \emph{global time} in $\Sscr$. We write $\values(\Sscr)$ to denote the set of all ground terms and timestamps occurring in $\Sscr$. Configurations are modified by \emph{multiset rewrite rules}. Only one rule, $\Tick$, modifies global time:
\begin{equation}
\label{eq:tick}
  \Time@T \lra \Time @ (T+1)
\end{equation}
where $T$ is a time variable. The $\Tick$ rule modifies a configuration to which it is applied by advancing the global time by one. The remaining rules are \emph{instantaneous} in that they do not modify the global time but may modify the remaining facts of a configuration. Instantaneous rules are given by expressions of the form
\begin{equation}
\begin{array}{l}
 \Time@T, \Wscr, 
 {F_1@T_1}, \ldots 
 {F_n@T_n} \mid \Cscr \\
 \quad \lra \Time@T, \Wscr, 
 {Q_1@(T + D_1)}, \ldots 
 {Q_m@(T + D_m)}
\label{eq:instantaneous}
\end{array}
\end{equation}
where $\Wscr$ (the \emph{side condition} of the rule) is a multiset of timestamped atomic formulas, $F_i@T_i$ is a timestamped atomic formula for each $i \leq n$, and $Q_j@(T+D_j)$ is a timestamped atomic formula for each $j \leq m$. The \emph{precondition} of the rule is the multiset $\{\Time@T\} \cup \Wscr \cup \{ F_i@T_i \mid i \leq n \}$, while its \emph{postcondition} is the multiset $\{\Time@T\} \cup \Wscr \cup \{ Q_j@(T+D_j) \mid j \leq m \}$. We require that no atomic formula in the multiset $\{ F_i@T_i \mid i \leq n \}$ appears with the same multiplicity as it appears in the multiset $\{ Q_j@(T+D_j) \mid j \leq m \}$. Furthermore, no timestamped atomic formulas containing the predicate $\Time$ can occur in $\{ F_i@T_i \mid i \leq n \} \cup \{ Q_j@(T+D_j) \mid j \leq m \}$. The \emph{guard} $\Cscr$ of the rule is a set of \emph{time constraints} of the form
$$T_1 > T_2 \pm N \quad \textrm{or} \quad T_1 = T_2 \pm N,$$
where $T_1$ and $T_2$ are time variables and $N\in\mathbb{N}$ is a natural number; all constraints in $\Cscr$ must involve only the time variables occurring in the rule's precondition.

A \emph{ground substitution} is a partial map $\sigma: \FOvar \cup \Tvar \to \GTerms \cup \mathbb{N}$ which maps first-order variables to ground terms and time variables to natural numbers. Given a multiset $W$ of timestamped atomic formulas, we write $W\sigma$ to denote the multiset of timestamped facts obtained by simultaneously substituting all first-order variables and time variables in $W$ with their image under $\sigma$. Given a set $\Cscr$ of time constraints with time variables from $W$, we say that $\Cscr \sigma$ is \emph{satisfied} if each time constraint in $\Cscr$ evaluates to true for the substituted timestamps. Given a multiset $W$ of timestamped atomic formulas, we write $\var(W)$ to denote the set of first-order variables and time variables occurring in $W$. Given an instantaneous rule $r$ given by $W \mid \Cscr \lra W'$, we write $\fresh(r)$ to denote the set $\var(W') \setminus \var(W)$.

A ground substitution matching an instantaneous rule $r$ given by $W \mid \Cscr \lra W'$ to a configuration $\Sscr$ is a ground substitution $\sigma$ with $\dom(\sigma) = \var(W \cup W')$ such that every element of $\var(W)$ is mapped to an element in $\values(\Sscr)$, and the restriction of $\sigma$ to $\fresh(r)$ is an injective map whose range is contained in $\Gcons \setminus \values(\Sscr)$. In other words, $\sigma$ assigns first-order variables (resp. time variables) occurring in $W$ to ground terms (resp. timestamps) occurring in $\Sscr$, and each distinct first-order variable in $\fresh(r)$ to a \emph{fresh} ground constant which does not occur in $\Sscr$.

An instantaneous rule $r$ given by $W \mid \Cscr \lra W'$ is \emph{applicable} to a configuration $\Sscr$ if there exists a ground substitution matching $r$ to $\Sscr$ such that $W\sigma \subseteq \Sscr$ and $\Cscr\sigma$ is satisfied; in this case, we refer to the expression $r\sigma$ given by $W\sigma \mid \Cscr\sigma \lra W'\sigma$ as an \emph{instance} of the rule $r$. The result of applying the rule instance $r\sigma$ to $\mathcal{S}$ is the configuration $(\Sscr \setminus W\sigma) \cup W'\sigma$. If $\Wscr$ is the side condition of $r$, and $T$ is the global time in $\Sscr$, then we say that the timestamped facts occurring in $(W \setminus (\Wscr \cup \{\Time@T\})) \sigma$ are \emph{consumed}, while those in $(W' \setminus (\Wscr \cup \{\Time@T\})) \sigma$ are \emph{created}. Note that a fact for the predicate $\Time$ is never created by an instantaneous rule. We write $\Sscr \lra_r \Sscr'$ for the one-step relation where the configuration $\Sscr$ is rewritten to $\Sscr'$ using an instance of the rule $r$. It is worth emphasizing, at this point, that configurations are \emph{grounded}, while rewrite rules are \emph{symbolic}.


\subsubsection*{Some examples.}
We now give some examples to elucidate our formalism. Consider the alphabet containing the predicate symbols $\Time$, $\At$, $\Event$, $\Attended$, and $\Flight_D$ (where $D \in \{1,\hdots,12\}$), and the constant symbols $\no$, $\done$, $\main$, $\airport$, $\Center$, $\id_{14}$, and $\id_{215}$. Recall that, in our expository example, we are modeling a researcher with a goal of traveling to attend a conference. We interpret the timestamped atomic formula $\Flight_D(id, c_1,c_2)@T$ to mean that the flight with flight id $id$ from city $c_1$ to city $c_2$ departs at time $T$ and has a duration of approximately $D$ hours.

The timestamped fact $\At(\FRA,\Center)@0$ is interpreted to mean that the researcher is at Frankfurt city center at the initial time step $0$. For this scenario, each time step is interpreted as the passage of one minute. For ease of readability, we adopt a more convenient representation of timestamps, with $0$ denoting midnight on the initial day of the planning scenario. Then, we write $\Time@(3d~\text{14:42})$ to indicate that the current time is 14:42 on the $3^{\text{rd}}$ day of the scenario. We do this is in lieu of writing the more burdensome timestamp $\Time@5202$. The fact $\Event(\main, \id_{215})@(5d~\text{12:00})$ specifies that the main event of the conference, with event identifier $215$, will take place at noon on the $5^{\text{th}}$ day. Bringing this all together, consider the following configuration.
\begin{equation}
\label{eq:ex-config1}
\begin{array}{l}
\{
\Time@(3d~\text{14:42}), \Attended(\main,\no)@0,
~\At(\FRA, \airport)@(3d~\text{14:05}),
\\
\Event(\main,\id_{215})@(5d~\text{12:00}), \Flight_2(\id_{14},\FRA, \DBV)@(3d~\text{15:25})
\} 
\end{array}
\end{equation}

This configuration describes a state of the system. The time is 14:42 on the $3^{\text{rd}}$ day of the scenario, and the researcher arrived at Frankfurt airport ($\FRA$) at 14:05. The main event of the conference is at noon in two days in Dubrovnik ($\DBV$), and has, obviously, not yet been attended by the researcher. Flight $\id_{14}$ is a direct flight from Frankfurt to Dubrovnik, which departs at 15:25 and has a duration of approximately two hours.

In addition to modeling states of the system via configurations, we also want our formalism to be able to model actions taken by the researcher, such as boarding a given flight. To this end, consider the rule
\begin{equation}
\label{eq:ex-rule}
\begin{array}{l}
\Time@T,  
\Flight_2(a,x,y)@T_1, \At(x, \airport)@T_2,
\, \mid \, T = T_1, T_2+30 \leq  T 
\\
\qquad \qquad \lra 
\Time@T, 
\Flight_2(a,x,y)@T_1, \At(y, \airport)@(T+120),
\end{array} 
\end{equation}
with side condition $\{ \Flight_2(a,x,y)@T_1 \}$. This rule means that if the departure time of a two-hour flight with flight id $a$ from city $x$ to city $y$ will depart at time $T$, and the researcher is at the airport in city $x$ at time $T_2$, where $T_2$ is at least $30$ minutes prior to $T$, then he can take the flight, arriving at the airport in city $y$ after two hours.

Note that the rule (Eq.~\ref{eq:ex-rule}) is not applicable to the configuration (Eq.~\ref{eq:ex-config1}). In particular, the time constraint $T=T_1$ cannot be satisfied by any ground assignment for the rule to the configuration. However, rule (Eq.~\ref{eq:ex-rule}) \emph{is} applicable to the configuration resulting from the successive application of 43 $\Tick$ rules to configuration (Eq.~\ref{eq:ex-config1}), which results in the same configuration, except with the timestamp for $\Time$ updated to $(3d~\text{15:25})$ (i.e., the departure time of the flight). Then the ground substitution $\sigma$ given by
\begin{align}
\sigma(T) &= 3d~\text{15:25} & \sigma(a) &= \id_{14} \notag \\
\sigma(T_1) &= 3d~\text{15:25} & \sigma(x) &= \FRA \notag \\
\sigma(T_2) &= 3d~\text{14:05} & \sigma(y) &= \DBV \notag
\end{align}
applied to the rule (Eq.~\ref{eq:ex-rule}) yields an instance which can be applied to configuration (Eq.~\ref{eq:ex-config1}), resulting in the following configuration:
\begin{equation}
\label{eq:ex-config2}
\begin{array}{l}
\{
\Time@(3d~\text{15:25}), \Attended(\main,\no)@0,
~\At(\DBV, \airport)@(3d~\text{17:25}),  \\
\Event(\main,\id_{215})@(5d~\text{12:00}), 
\Flight_2(\id_{14},\FRA, \DBV)@(3d~\text{15:25})
\}.
\end{array}
\end{equation}

\subsubsection*{Timed MSR Systems.}
We now turn to the timed MSR systems introduced in \cite{kanovich17mscs}.

\begin{definition}
A \emph{timed MSR system} 
$\Ascr$ is a set of rules containing only instantaneous rules (Eq.~\ref{eq:instantaneous}) and the $\Tick$ rule (Eq.~\ref{eq:tick}).
\end{definition}

A sequence of consecutive rule applications represents an execution or process within the system. A \emph{trace} of timed MSR rules $\Ascr$ starting from an initial configuration $\Sscr_0$ is a sequence of configurations:
\(
\Sscr_0 \lra \Sscr_1 \lra \Sscr_2 \lra \cdots \lra \Sscr_n 
\),
such that for all $0 \leq i \leq n-1$, $\Sscr_{i} \lra_{r_i} \Sscr_{i+1}$ for some $r_i \in \Ascr$. For our complexity results, we assume traces are annotated with the rule instances used to obtain the next configuration in the trace, so valid traces can be recognized in polynomial time (cf. Remark \ref{remark:poly-verification}).

Reachability problems for MSR systems are to determine whether or not a trace from some initial configuration to some specified configuration exists. In general, these problems are often undecidable\cite{kanovich17mscs}, and so restrictions are imposed in order to obtain decidability\footnote{For a discussion of various conditions in the model that may affect complexity, see~\cite{kanovich17jcs,kanovich17mscs}.}.
In particular, we use MSR systems with only \emph{balanced} rules.

\begin{definition}[Balanced Rules, \cite{kanovich17mscs}]
\label{def:balanced}
A timed MSR  rule is \emph{balanced} if the numbers of facts on left and right sides of the rule are equal.
\end{definition}

\noindent Systems containing only balanced rules represent an important class of \emph{balanced systems}, for which several reachability problems have been shown to be decidable~\cite{kanovich17jcs}. Balanced systems are suitable, \eg, for modeling scenarios with a fixed amount of total memory. Balanced systems have the following important property:
\begin{proposition}[\cite{kanovich17jcs}]
\label{prop:balanced-config}
Let $\Rscr$ be a set of balanced rules. Let $\Sscr_0$ be a configuration with exactly $m$ facts (counting multiplicities). Let ~$\Sscr_0 \lra \cdots \lra \Sscr_n$~ be an arbitrary trace of $\Rscr$ rules starting from $\Sscr_0$. Then for all ~$0 \leq i \leq n$, ~$\Sscr_i$ has exactly $m$ facts.
\end{proposition}

\subsection{Progressing Timed Systems} 
\label{sec:pts}

In this section, we review a particular class of timed MSR systems, called \emph{progressing timed MSR systems} (PTSs)~\cite{kanovich16formats,kanovich2021Guttman}, in which only a bounded number of rules can be applied in a single time step. This is a natural condition, similar to the \emph{finite-variability assumption} used in the temporal logic and timed automata literature~\cite{hirshfeld2004logics}.

\begin{definition}[Progressing Timed System, \cite{kanovich16formats}]
\label{def:progressing} \ 
An instantaneous rule $r$ of the form in (Eq.~\ref{eq:instantaneous}) is \emph{progressing} if the following all hold:
\begin{inparaenum}[i)]
\item $n = m$ (i.e., $r$ is balanced);
\item $r$ consumes \emph{only} facts with timestamps in the past or at the current time, \ie, in (Eq.~\ref{eq:instantaneous}), the set of constraints ~$\Cscr$ of $r$ contains the set
~$\Cscr_r = \{~T \geq T_i \mid F_i@T_i, ~1 \leq i \leq n~\};$
\item $r$ creates \emph{at least one} fact with timestamp greater than the global time, \ie, in (Eq.~\ref{eq:instantaneous}), ~ $D_i \geq 1$~ for at least one ~$i \in \{1, \dots, n \}$.
\end{inparaenum}
A timed MSR system $\Ascr$ is a \emph{progressing timed MSR system (PTS)} if all instantaneous rules of $\Ascr$ are progressing.
\end{definition}

Note that the rule (Eq.~\ref{eq:ex-rule}) is progressing. A timestamped fact in a configuration $\Sscr$ is a \emph{future fact} if its timestamp is strictly greater than the timestamp of the $\Time@t$ fact in $\Sscr$. Future facts are ``not available'' in the sense that they cannot be consumed by a progressing rule before a sufficient number of $\Tick$ rule applications.

\begin{remark}
\label{remark: pts}
For readability, we assume the set of constraints for all rules $r$, contains the set $\Cscr_r$, as in Definition~\ref{def:progressing}, and do not always write $\Cscr_r$ explicitly.
\end{remark}

\subsection{Timed MSR for the specification of resilient systems}
\label{sec: msr-resilience}
We now review additional notation for the purpose of specifying resilience, as introduced in \cite{Alturki2022resilience}. The resilience framework divides the system from an external entity, such as the environment, regulatory authorities, or an adversary. We model various types of disruptive changes to the system state or goals.

\begin{definition}[Planning Configuration, \cite{Alturki2022resilience}]
\label{def:planning-config}
Let $\Sigma_P = \Sigma_G \uplus \Sigma_C \uplus \Sigma_S \uplus \{\Time\}$ consist of four pairwise disjoint sets of predicate symbols, $\Sigma_G$,  $\Sigma_C$,  $\Sigma_S$ and  $\{\Time\}$. Facts constructed using predicates from $\Sigma_G$ are called \emph{goal facts}, from $\Sigma_C$ \emph{critical facts}, and from $\Sigma_S$ \emph{system facts}. Facts constructed using predicates from $\Sigma_C \cup\Sigma_G$ are called \emph{planning facts}. Configurations over $\Sigma_P$ predicates are called \emph{planning configurations}.
\end{definition}

For readability, we underline predicates in planning facts and refer to planning configurations as configurations for short. The behavior of the system is represented by traces of MSR rules. A system should achieve its goal while not violating predetermined critical conditions. This is made precise in the following two definitions.

\begin{definition}[Critical/Goal Configurations, \cite{Alturki2022resilience}]
\label{def:critical-goal-config}
A \emph{critical (resp. goal) configuration specification} $\CS$ (resp. $\GS$) is a set of pairs 
$\{\tup{\Sscr_1, \Cscr_1}, \ldots, \tup{\Sscr_n, \Cscr_n}\}$,  with each pair $\tup{\Sscr_j,\Cscr_j}$ being of the form
$\tup{\{F_1@T_1, \ldots, F_{p_j}@T_{p_j}\}, \Cscr_j}$, 
where $T_1, \ldots, T_{p_j}$ are time variables, $W = \{ F_1, \ldots, F_{p_j} \}$ is a multiset of timestamped atomic formulas, with at least one occurrence of a critical (resp. goal) predicate symbol, and $\Cscr_j$ is a set of time constraints involving only variables $T_1, \ldots, T_{p_j}$. A configuration $\Sscr$ is a {\em critical configuration} w.r.t. $\CS$ (resp. a {\em goal configuration} w.r.t. $\GS$)
if for some $1 \leq i \leq n$, there is a grounding substitution $\sigma$ with $\dom(\sigma) = \var(W)$ such that $\Sscr_i \sigma \subseteq \Sscr$ and $\Cscr_i \sigma$ is satisfied.
\end{definition}

\begin{definition}[Compliant Traces, \cite{Alturki2022resilience}]
\label{def:compliant}
A trace is \emph{compliant} with respect to a critical configuration specification $\CS$ if it does not contain any critical configuration w.r.t. $\CS$.
\end{definition}

Note that critical configuration specifications and goal configuration specifications, like rewrite rules, are \emph{symbolic}. Reaching a critical configuration may be interpreted as a \emph{safety violation}, while a compliant trace may be interpreted as a \emph{safe trace}. As an example, suppose that in the example alphabet introduced earlier, the predicate symbol $\Attended$ is in $\Sigma_C$, while the predicate symbol $\Event$ is in $\Sigma_G$. Then the goal configuration specification
$$\{\tup{\{ \Attended(\main, \done)@T_1, \Event(\main, x)@T_2 \}, \emptyset }\}$$
indicates that the main event should be attended, while the critical configuration specification 
\begin{align*}
\{\langle \, 
 \Time@T, \Attended(\main, \no)@T_1, \Event(\main, x)@T_2 \}, \{ T > T_2 \} \rangle\}
\end{align*}
denotes that it is critical not to participate in the main event. 

\begin{definition}[System Rules and Update Rules, \cite{Alturki2022resilience}]
Fix a planning alphabet $\Sigma_P$. A \emph{system rule} is either the $\Tick$ rule (Eq.~\ref{eq:tick}) or a rule of form in (Eq.~\ref{eq:instantaneous}) which does not consume or create planning facts. An \emph{update rule} is an instantaneous rule that is of one of the following types: (a) a \emph{system update rule (\SUR)} such that planning facts can only occur in the side condition of the rule; or (b) a \emph{goal update rule (\GUR)} that either consumes or creates at least one goal fact and such that critical facts can only occur in the side condition of the rule.
\end{definition}
For example, the following system rule specifies that the traveler needs 40 
minutes to get from the departing city center to the airport:
\begin{align*}
\Time@T, ~&\At(x, \Center)@T_1 \, \mid \, T_1 \leq T  \\
\lra ~&\Time@T, \At(x, \airport)@(T+40).
\end{align*}

The rule (Eq.~\ref{eq:ex-rule}) is another example of a system rule. System rules specify the behavior of the system, while disruptions are modeled via update rules. Intuitively, \GUR 
model external interventions in the system, such as mission changes, additional tasks, etc., while \SUR model changes in the system that are not due to the intentions of the system's agents,  \eg, technical errors or malfunctions such as flight delays. Both goal and system update rules can create and/or consume system facts, which technically simplifies modeling the impact of changes on the system and its response. For example, the following \GUR models a change in the scheduled time of the main event.
\begin{align*}
\Time@T, ~&\Event(\main,x)@T_1, 
\lra \Time@T, \Event(\main,x)@(T+60),
\end{align*}
while the following \SUR models a 30-minute flight delay:
\begin{align*}
    \Time@T, \Flight_D(a,x,y)@T_1  \lra \Time@T,   \Flight_D(a,x,y)@(T+30). 
\end{align*}

\begin{definition}[Planning Scenario, \cite{Alturki2022resilience}]
If $\Rscr$ and $\Escr$ are sets of system and update rules, $\GS$ and $\CS$ are a goal and critical configuration specifications, and $\Sscr_0$ is an initial configuration, then the tuple $(\Rscr,\GS,\CS,\Escr,\Sscr_0)$ is a \emph{planning scenario}.
\end{definition}

\begin{definition}[Progressing Planning Scenario (PPS)]
\label{def:progressing scenario}
We say that a planning scenario $(\Rscr,\GS,\CS,\Escr,\Sscr_0)$ is \emph{progressing} if all rules in $\Rscr$ and $\Escr$ are progressing.
\end{definition}

The progressing condition in Definition \ref{def:progressing scenario} implies a bound on the number of rules that can be applied in a single unit of time (cf. Proposition \ref{prop:bounded-length}). We also assume an upper-bound on the size of facts allowed to occur in traces, where the size of a timestamped fact $F@t$ is the number of symbols from $\Sigma$ occurring in $F$, counting repetitions. Without this bound (among other restrictions), many interesting decision problems are undecidable~\cite{durgin04jcs,kanovich17jcs}. We also confine attention to classes of \emph{$\varmax$-simple} PPSs.

\begin{definition}
\label{def:bounded-PTS}
Let $\varmax$ denote a fixed positive integer. We say that a planning scenario $A=(\Rscr,\GS,\CS,\Escr,\Sscr_0)$ is \emph{$\varmax$-simple} if the total number of variables (including both first-order and time variables) appearing in each pair $\langle \Sscr_i, \Cscr_i \rangle$ in $\CS$ is less than $\varmax$.
\end{definition}

\noindent For every planning scenario $A=(\Rscr,\GS,\CS,\Escr,\Sscr_0)$, there exists some least $\varmax$ such that $A$ is $\varmax$-simple; intuitively, this $\varmax$ is a measure of the complexity of verifying compliance of traces with respect to $\CS$. Proposition \ref{prop:BPTS-poly} in Section \ref{sec:complexity-b-time} makes this intuition precise.

\begin{remark}\label{remark: example-pts}
By inspecting the rules and the critical configuration specification, it is easy to check that our expository travel example is 3-simple and progressing.
\end{remark}

\section{Time-bounded Resilience Verification Problems}
\label{sec:bounded}
In this section, we formalize time-bounded resilience as a property of planning scenarios. Intuitively, we want to capture the notion of a system which can achieve its goal within a fixed amount of time, despite the application of up to $n$ instances of update rules. An initial idea might be to require that the system can achieve its goal in the allotted time, regardless of when updates are applied. However, this is too restrictive: many systems will fail to achieve their goal in the face of adversarial actions which can be applied arbitrarily often. Instead, the system will have $a+b$ time units to achieve its goal, and update rules can only be applied in the first $a$ time steps; the last $b$ time steps are the \emph{recovery time} afforded to the system.

\begin{definition}[The $(n,a,b)$-resilience problem]
\label{d-trace-ab-recursive}
Let $a \in \mathbb{Z}^+$ and $b \in \mathbb{N}$. We define $(n,a,b)$-resilience by recursion on $n$. Inputs to the problem are planning scenarios $A=(\Rscr,\GS,\CS,\Escr,\Sscr_0)$. A trace is \emph{$(0,a,b)$-resilient with respect to $A$} if it is a compliant trace of $\Rscr$ rules from $\Sscr_0$ to a goal configuration and contains at most $a+b$ applications of the $\Tick$ rule. Let $t_0$ denote the global time in the configuration $\Sscr_0$. For $n > 0$, a trace $\tau$ is \emph{$(n,a,b)$-resilient with respect to $A$} if
\begin{enumerate}
\item $\tau$ is $(0,a,b)$-resilient with respect to $A$, and
\item for any system or goal update rule $r \in \Escr$ such that $\Sscr_i \lra_{r} \Sscr'_{i+1}$ for some configuration $\Sscr_i$ in $\tau$ with global time $t_i$, where $d_i = t_i - t_0 \leq a$, there is a \emph{reaction trace} $\tau'$ of $\Rscr$ rules from $\Sscr'_{i+1}$ to a goal configuration $\Sscr'$ such that $\tau'$ is $(n-1,a-d_i,b)$-resilient with respect to $A' = (\Rscr,\GS,\CS,\Escr,\Sscr'_{i+1})$.
\end{enumerate}
A planning scenario $A = (\Rscr,\GS,\CS,\Escr,\Sscr_0)$ is \emph{$(n,a,b)$-resilient} if an $(n,a,b)$-resilient trace with respect to $A$ exists. The $(n,a,b)$-resilience problem is to determine if a given planning scenario $A$ is $(n,a,b)$-resilient.
\end{definition}

\noindent Figure \ref{fig:resilience} provides a visual depiction of Definition \ref{d-trace-ab-recursive}.

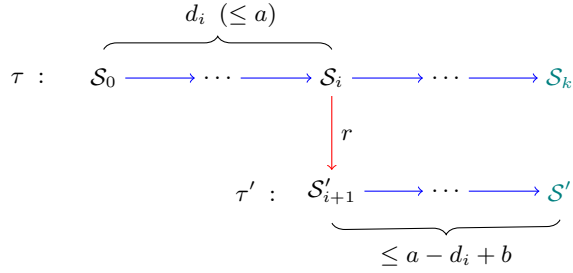
\begin{figure}[H]
\centering
\begin{tikzpicture}
  \node (tau) at (-1,0) {$\tau~:$};
  \node (S0) at (0,0) {$\Sscr_0$};
  \node (dots1) at (1.5,0) {$\hdots$};
  \node (Si) at (3,0) {$\Sscr_i$};
  \node (dots2) at (4.5,0) {$\hdots$};
  \node (Sk) at (6,0) {$\textcolor{teal}{\Sscr_k}$};

  \draw[->,blue] (S0) -- (dots1);
  \draw[->,blue] (dots1) -- (Si);
  \draw[->,blue] (Si) -- (dots2);
  \draw[->,blue] (dots2) -- (Sk);

  \node (tau') at (2,-1.5) {$\tau'~:$};
  \node (Si1) at (3,-1.5) {$\Sscr'_{i+1}$};
  \node (dots3) at (4.5,-1.5) {$\hdots$};
  \node (Sp) at (6,-1.5) {$\textcolor{teal}{\Sscr'}$};
  \node (r) at (3.2,-0.75) {$r$};

  \draw[->,red] (Si) -- (Si1);
  \draw[->,blue] (Si1) -- (dots3);
  \draw[->,blue] (dots3) -- (Sp);

  \draw [decorate,decoration={brace,amplitude=5pt,mirror,raise=4pt},yshift=0pt]
    (Si.north) -- (S0.north) node [black,midway,xshift=5pt,yshift=17pt] {$d_i ~~(\leq a)$};

  \draw [decorate,decoration={brace,amplitude=5pt,mirror,raise=4pt},yshift=0pt]
    (Si1.south) -- (Sp.south) node [black,midway,xshift=0pt,yshift=-17pt] {$\leq a - d_i + b$};
\end{tikzpicture}
\caption{An $(n,a,b)$-resilient trace $\tau$ and an $(n-1,a-d_i,b)$-resilient reaction trace $\tau'$. The horizontal arrows correspond to system rule applications, while the downward-facing arrow represents an update rule application. The configurations $\Sscr_k$ and $\Sscr'$ on the far right are goal configurations.}
\label{fig:resilience}
\end{figure}

The reaction trace $\tau'$ in Definition \ref{d-trace-ab-recursive} can be interpreted as a change in the plan $\tau$, made in response to an external disruption (i.e., the system/goal update rule $r$) imposed on the system. Note that it is this ``replanning'' aspect of our definition that intuitively distinguishes it from the related notion of robustness.

\begin{remark}
\label{remark:parameter-observations}
In Definition \ref{d-trace-ab-recursive}, the global time $t'$ in $\Sscr'$ satisfies $t' - t_0 \leq a + b$; \ie, despite the application of $n$ instances of update rules, an $(n,a,b)$-resilient trace reaches a goal within $a+b$ time units. Furthermore, observe that a trace is $(n,a,b)$-resilient with respect to a planning scenario $A$ if and only if it is $(n,a,b')$-resilient with respect to $A$ for all $b' \geq b$. Similarly, all $(n,a,b)$-resilient traces with respect to $A$ are $(n',a,b)$-resilient with respect to $A$ for all $n' \leq n$.
\end{remark}

It is worthwhile to note that Definition \ref{d-trace-ab-recursive} can be seen as a modification of \cite[Definitions 9-10]{Alturki2022resilience}, in which $(i)$ we include the parameters $a$ and $b$, $(ii)$ we consider both system and goal update rules simultaneously, and $(iii)$ the \emph{recoverability condition}, which is not mentioned in this work, is the total relation on configurations of $A$.

\section{Computational Complexity of Time-Bounded Resilience}
\label{sec:complexity-b-time}
In this section, we state and prove our results on the computational complexity of the time-bounded resilience problem defined in Section \ref{sec:bounded}. To see this, we first state a known bound on the number of instances of instantaneous rules appearing between two consecutive instances of $Tick$ rules in a trace of only progressing rules.

\begin{proposition}[\cite{kanovich16formats}]
\label{prop:bounded-length}
Let $\Rscr$ be a set of progressing rules, $\Sscr_0$ an initial configuration, and $m$ the number of facts in $\Sscr_0$. For all traces $\tau$ of $\Rscr$ rules starting from $ \Sscr_0$, let
$$\Sscr_i \lra_{Tick}  \Sscr_{i+1} \lra \cdots \lra \Sscr_j \lra_{Tick}  \Sscr_{j+1}$$
be any subtrace of ~$\tau$ with exactly two instances of the $Tick$ rule, one at the beginning and the other at the end. Then ~$j - i \leq m$.
\end{proposition}

\noindent Proposition \ref{prop:bounded-length} guarantees that the lengths of $(n,a,b)$-resilient traces of a PPS $A$ are polynomially-bounded in the size of the input representation of $A$.

\begin{proposition}
\label{prop:small-traces}
Let $A=(\Rscr,\GS,\CS,\Escr,\Sscr_0)$ be a PPS and $m$ be the number of facts in $\Sscr_0$. Then the length of any $(n,a,b)$-resilient trace of $A$ is bounded by ~$(a+b+1)m$.
\end{proposition}

\ifthenelse{\boolean{longversion}}{
In preparation for our $(n,a,b)$-resilience upper bound result (Theorem \ref{thm:upper-bound}), we now turn to the complexity of some fundamental decision problems pertaining to planning scenarios: the \emph{goal/critical recognition problems}, the \emph{trace compliance problem}, and the \emph{rule applicability problem}. The complexity of these problems in the case of arbitrary planning scenarios stems from the complexity of finding appropriate ground substitutions, which are closely related to homomorphisms between finite relational structures.

\begin{definition}
\label{def:goal-rec-prob}
The \emph{goal (resp. critical) recognition problem} is to determine, given a planning scenario $A=(\Rscr,\GS,\CS,\Escr,\Sscr_0)$ and a configuration $\Sscr$, whether or not $\Sscr$ is a goal (resp. critical) configuration w.r.t. $\GS$ (resp. $\CS$); cf. Definition \ref{def:critical-goal-config}.
\end{definition}

\begin{proposition}
\label{prop:goal-rec-npc}
The goal recognition problem is $\NP$-complete.
\end{proposition}
\begin{proof}
For membership in $\NP$, suppose we are given an input consisting of a planning scenario $A=(\Rscr,\GS,\CS,\Escr,\Sscr_0)$ and a configuration $\Sscr$. As witnesses, we use ground substitutions $\sigma$ from the collection of variables occurring in $\Sscr_j$ for some pair $\langle \Sscr_j, \Cscr_j \rangle$ in $\GS$ to the collection of ground terms occurring in $\Sscr$. Note that each witness of this form is clearly polynomially-bounded with respect to the size of the input representation of $A$. The verification algorithm is to iterate through each pair $\langle \Sscr_i, \Cscr_i \rangle$ in $\GS$, applying the substitution $\sigma$ to each and checking if $\Sscr_i \sigma \subseteq \Sscr$ and $\Cscr_i \sigma$ is satisfied in $\Sscr$. If, for any pair, these conditions are met, then we accept. Otherwise, we reject. This algorithm clearly runs in polynomial time, and so the goal recognition problem is in $\NP$.

Furthermore, the goal recognition problem is $\NP$-hard, by a reduction from the $\NP$-complete graph homomorphism problem \cite{garey1979computers}. To see this, suppose we are given graphs $G = (V,E)$ and $K = (V',E')$. Fix a bijection from $G$ to a set of first-order variables, and a bijection from $K$ to a set of ground terms; let $f$ denote the disjoint union of these two mappings. Consider the signature of one binary relation symbol $R$, and set
\begin{align*}
\Sscr_0 &= \{ R(x,y)@T_{xy} \mid x = f(u) ~\text{and}~ y = f(v) ~\text{for some}~ (u,v) \in E \}, ~\text{and} \\
\Sscr &= \{ \Time@0 \} \cup \{ R(a,b)@0 \mid a = f(u) ~\text{and}~ b = f(v) ~\text{for some}~ (u,v) \in E' \}.
\end{align*}
Consider the planning scenario $A=(\Rscr,\GS,\CS,\Escr,\Sscr_0)$, where $\Rscr$ contains only the $\Tick$ rule, $\CS = \Escr = \emptyset$, and $\GS = \{ \langle \Sscr_0, \emptyset \rangle \}$. Clearly $A$ can be computed in polynomial time with respect to the input representation of the pair $(G,K)$, and so it suffices to show that a homomorphism from $G$ to $K$ exists if and only if $\Sscr$ is a goal configuration.

For the forward direction, suppose a homomorphism $h: G \to K$ exists. Let $\sigma$ be the ground substitution sending each variable $x$ in $\Sscr_0$ to $f(h(u))$, where $u$ is the unique element of $G$ such that $x = f(u)$. Now suppose that a fact $R(x',y')$ is in $\Sscr_0 \sigma$. Then there exists a fact $R(x,y)$ in $\Sscr_0$ such that $x' = \sigma(x)$ and $y' = \sigma(y)$. Hence there exists a pair $(u,v) \in E$ such that $x = f(u)$ and $y = f(v)$. Then because $h$ is a homomorphism, we have that $(h(u),h(v)) \in E'$. Hence $R(f(h(u)),f(h(v))) = R(\sigma(x), \sigma(y)) = R(x',y')$ is in $\Sscr$. Therefore, we conclude that $\Sscr_0 \sigma \subseteq \Sscr$. Since the empty set of constraints is satisfied in $\Sscr$ as well, we conclude that $\Sscr$ is a goal configuration.

For the reverse direction, suppose that $\Sscr$ is a goal configuration. Then there exists a ground substitution $\sigma$ such that $\Sscr_0 \sigma \subseteq \Sscr$. Let $h: G \to K$ be the map sending $u$ to $\sigma(f(u))$. Then if $(u,v) \in E$, we have that $R(f(u),f(v)) \in \Sscr_0$. Then we have that $R(\sigma(f(u)),\sigma(f(v))) \in \Sscr_0 \sigma$, and since $\Sscr_0 \sigma \subseteq \Sscr$, it must be the case that $R(\sigma(f(u)),\sigma(f(v))) \in \Sscr$. Hence $(\sigma(f(u)),\sigma(f(v))) = (h(u), h(v)) \in E'$, and so we conclude that $h$ is a homomorphism. \qed
\end{proof}

\noindent A nearly identical proof also yields the following proposition.

\begin{proposition}
\label{prop:crit-rec-npc}
The critical recognition problem is $\NP$-complete. Consequently, it is $\coNP$-complete to determine if a configuration is \emph{not} critical.
\end{proposition}

\noindent More broadly, we are interested in checking trace compliance.

\begin{definition}
\label{def:trace-compliance-prob}
The \emph{trace compliance problem} is to determine, given a planning scenario $A = (\Rscr,\GS,\CS,\Escr,\Sscr_0)$ and a trace $\tau$ of $\Rscr$-rules starting from $\Sscr_0$, whether or not $\tau$ is compliant w.r.t. $\CS$ (cf. Definition \ref{def:compliant}).
\end{definition}

Note that in Definition \ref{def:trace-compliance-prob}, the trace $\tau$ in the input could contain a single configuration, and so the trace compliance problem is $\coNP$-hard by a trivial reduction from the complement of the critical recognition problem. Furthermore, a trace $\tau$ is compliant if and only if no ground substitution witnesses that some configuration in $\tau$ is critical. This leads easily to the following proposition.

\begin{proposition}
\label{def:trace-compliance-conpc}
The trace compliance problem is $\coNP$-complete.
\end{proposition}

For the remainder of this section, we fix an arbitrary constant $\eta$. If $A$ is an \emph{$\varmax$-simple} PPS (cf. Definition \ref{def:bounded-PTS}), then given a configuration $\Sscr$ with $m$ facts, there are at most $m^\varmax$ possible ground substitutions witnessing that $\Sscr$ is critical (or is a goal). It follows that, for $\varmax$-simple planning scenarios, even a brute force search through ground substitutions can be used to solve the trace compliance problem in polynomial time, yielding the following proposition.

\begin{proposition}
\label{prop:BPTS-poly}
For $\eta$-simple planning scenarios, the trace compliance problem is in $\P$.
\end{proposition}

\noindent The final fundamental problem for planning scenarios is the following.

\begin{definition}
\label{def:rule-app-prob}
The \emph{rule applicability problem} is to determine, given a planning scenario $A=(\Rscr,\GS,\CS,\Escr,\Sscr_0)$, a rule $r \in \Rscr \cup \Escr$, and a configuration $\Sscr$, whether or not $r$ is applicable to $\Sscr$.
\end{definition}

\begin{proposition}
\label{prop:rule-app-npc}
The rule applicability problem is $\NP$-complete.
\end{proposition}
\begin{proof}[sketch]
The proof is by a slight modification of the proof of Proposition \ref{prop:goal-rec-npc}. In this case, given graphs $G$ and $K$, we encode the graph $G$ into the precondition of the rule $r$, and we encode the graph $K$ into the configuration $\Sscr$. \qed
\end{proof}

\begin{remark}
\label{remark:poly-verification}
The witnesses used to show that these problems are in $\NP$ and $\coNP$ are all ground substitutions. Let $A=(\Rscr,\GS,\CS,\Escr,\Sscr_0)$ be an $\varmax$-simple PPS, and $\Sscr$ a configuration with the same number of facts as $\Sscr_0$. Given an appropriate ground substitution, we can verify in polynomial time in the size of $A$ that $\Sscr$ is a goal configuration w.r.t. $\GS$. In the proof of Theorem \ref{thm:upper-bound}, for ease of exposition, we will also assume that these ground substitutions come with a pointer to the appropriate pair $\langle \Sscr_i, \Cscr_i \rangle$ in $\GS$ to which the substitution should be applied. Similarly, given an appropriate ground substitution, we can verify in polynomial time in the size of $A$ that a rule $r$ is applicable to $\Sscr$, and whether or not $\Sscr'$ is the result of this application.
\end{remark}
}{

}

We now turn our attention to the computational complexity of the $(n,a,b)$-resilience problem. To establish our complexity results, we will utilize the quantifier-alternation characterization of $\PH$ (cf. \cite{arora2009complexity,stockmeyer1976polynomial,papadimitriou07book}), according to which a decision problem is in $\Sigma^\P_n$ (for $n$ odd) if and only if there exists a polynomial-time algorithm $M$ such that an input $x$ is a \emph{yes} instance of the problem if and only if
$$\exists u_1 \forall u_2 \exists u_3 \hdots \forall u_{n-1} \exists u_n ~M(x,u_1,\hdots,u_n) ~\text{accepts},$$
where the $u_i$ are polynomially-bounded in the size of $x$. We now establish an upper bound on the complexity of the $(n,a,b)$-resilience problem (cf. Definition~\ref{d-trace-ab-recursive}).

\begin{theorem}
\label{thm:upper-bound}
For $\eta$-simple PPSs with traces containing only facts of bounded size and all $a \in \mathbb{Z}^+$ and $n, b \in \mathbb{N}$, there exists a decision procedure of complexity $\Sigma^\P_{2n+1}$ for the $(n,a,b)$-resilience problem.
\end{theorem}
\ifthenelse{\boolean{longversion}}{
\begin{proof}
We show by induction on $n$ that, for each $a \in \mathbb{Z}^+$ and $b \in \mathbb{N}$, there exists a polynomial-time algorithm $M^{a,b}_n$ such that an $\varmax$-simple PPS $A=(\Rscr,\GS,\CS,\Escr,\Sscr_0)$ is $(n,a,b)$-resilient if and only if
\begin{equation}
\label{eq:pi-eq}
\exists \Tscr_0 \forall \rho_1 \exists \Tscr_1 \hdots \forall \rho_n \exists \Tscr_n ~M^{a,b}_n(A,\tau_0,\tau_1,\hdots,\tau_n,\rho_1,\hdots,\rho_n) ~\text{accepts}.
\end{equation}
The existentially quantified variables $\Tscr_i$ range over triples of the form $(\tau_i,\sigma_i,j_i)$, where $\tau_i$ is a trace of $\Rscr$-rules and $\sigma_i$ is a ground substitution from $\langle \Sscr_{j_i}, C_{j_i} \rangle$ in $\GS$ to the last configuration of $\tau_i$. By Proposition \ref{prop:small-traces}, such a witness $\Tscr_i$ is polynomially-bounded in the size of the input representation of $A$. The universally quantified variables $\rho_i$ range over triples of the form $(r_i,\sigma_i,j_i)$, where $r_i \in \Escr$ and $\sigma_i$ is a ground substitution from the $j_i^{th}$ configuration of $\tau_{i-1}$ to the first configuration of $\tau_i$. The witnesses $\rho_i$ are also clearly polynomially-bounded in the size of the input representation of $A$.

For the base case, the algorithm $M^{a,b}_0$ first verifies that $A$ meets the syntactic requirements of an $\varmax$-simple PPS. If so, then we verify, given $\Tscr_0 = (\tau_0,\sigma_0,j_0)$, that $\tau_0$ has at most $a+b$ applications of the $Tick$ rule, is compliant, and leads to a goal. By Proposition \ref{prop:BPTS-poly}, since $\tau_0$ is polynomially-bounded in $A$, we can verify compliance of $\tau_0$ in polynomial time in $A$. By Remark \ref{remark:poly-verification}, we can verify in polynomial time in $A$, given $(\sigma_0,j_0)$, that the last configuration of $\tau_0$ is a goal. Hence $M^{a,b}_0(A,\Tscr_0)$ runs in polynomial time, and $A$ is $(0,a,b)$-resilient if and only if $\exists \Tscr_0 ~M^{a,b}_0(A,\Tscr_0)$ accepts.

Now suppose inductively that we have, for each $a' \in \mathbb{Z}^+$ and $b' \in \mathbb{N}$, algorithms $M^{a',b'}_k$ satisfying (Eq.~$\ref{eq:pi-eq}$) with $n = k$. Fix some $a \in \mathbb{Z}^+$ and $b \in \mathbb{N}$, and we define an algorithm $M^{a,b}_{k+1}$ which takes inputs of the form $(A,\Tscr,\Tscr',\Tscr_1,\hdots,\Tscr_k,\rho,\rho_1,\hdots,\rho_k)$. Let $\Tscr = (\tau,\sigma,j)$, $\Tscr' = (\tau',\sigma',j')$, and $\rho = (r,\sigma^\ast,i)$. Furthermore, let $t_0$ denote the global time in the initial configuration $\Sscr_0$, $\lvert \tau \rvert$ denote the length of $\tau$, $\mathcal{S}'_{i+1}$ denote the initial configuration of $\tau'$, $t_i$ denote the global time in the $i^{th}$ configuration $\mathcal{S}_i$ of $\tau$, and $d_i = t_i - t_0$. We now describe the run of $M^{a,b}_{k+1}$ on this input.

First, check that $\tau$ and $\tau'$ are compliant traces to a goal configuration. Then, check if $d_i \leq a$; if this check fails, then we halt and accept, since by Definition \ref{d-trace-ab-recursive}, update rules cannot be applied after more than $a$ time steps. Then, check if $\Sscr_i \lra_r \Sscr'_{i+1}$, by applying the ground substitution $\sigma^\ast$ to $r$ and checking that it is applicable to $\Sscr_i$. If this checks fails, then we halt and accept, since $r$ is not an applicable update rule to $\Sscr_i$. Otherwise, check that $\Sscr'_{i+1}$ is the correct result of applying this instance of $r$ to $\Sscr_i$. If this check fails, then reject, since $\tau'$ cannot be a valid reaction trace. Finally, let $A' = (\Rscr,\GS,\CS,\Escr,\Sscr'_{i+1})$, and simulate $M^{a-d_i,b}_k$ on the input $(A',\Tscr',\Tscr_1,\hdots,\Tscr_k,\rho_1,\hdots,\rho_k)$. If the result of this simulation is that $M^{a-d_i,b}_k$ accepts the input, then we halt and accept, since by the inductive hypothesis, $\tau'$ must be a $(k,a-d_i,b)$-resilient reaction trace. Otherwise, we reject.

Taking into account the inductive hypothesis and Remark \ref{remark:poly-verification}, it is clear that $M^{a,b}_{k+1}$ runs in polynomial time in the size of its input. Furthermore, it follows immediately by inspection of Definition \ref{d-trace-ab-recursive} that $A$ is $(k+1,a,b)$-resilient if and only if
$$\exists \Tscr_0 \forall \rho_1 \exists \Tscr_1 \hdots \forall \rho_{k+1} \exists \Tscr_{k+1} ~M^{a,b}_{k+1}(A,\Tscr,\Tscr_1,\hdots,\Tscr_{k+1},\rho,\rho_1,\hdots,\rho_{k+1}) ~\text{accepts}.$$
This concludes the inductive argument. It follows immediately from the quantifier-alternation characterization of $\PH$ that the $(n,a,b)$-resilience problem for $\eta$-simple PPSs with traces containing only facts of bounded size is in $\Sigma^\P_{2n+1}$. \qed
\end{proof}
}{

}

\begin{remark}
Even without assuming $\eta$-simplicity, a slight variation of the above argument gives a decision procedure of complexity $\Sigma^\P_{2n+2}$ for the $(n,a,b)$-resilience problem for PPSs with traces containing facts of bounded size. To modify the argument, we allow each universal quantifier to range over an additional ground substitution, which is used in the verification algorithm $M^{a,b}_n$ to check that an arbitrary configuration in the preceding witness trace is non-critical. Note that this check can be done in polynomial time (cf. Remark \ref{remark:poly-verification}). If this check succeeds for \emph{all} configurations and \emph{all} such ground substitutions, then every witness trace is compliant.
\end{remark}

\ifthenelse{\boolean{longversion}}{
\begin{remark}
In~\cite{kanovich16formats}, we showed that for PPSs with traces containing facts of bounded size, the $n$-time-bounded-realizability problem is in $\NP$, when it is assumed that compliance can be checked in polynomial-time. Note that by Proposition~\ref{def:trace-compliance-conpc}, this task is in fact $\coNP$-complete in general. The requirement in the above proof that we confine attention to $\varmax$-simple PPSs is one way to make this assumption precise; furthermore, $\varmax$-simplicity is easy to verify in polynomial time in the size of a planning scenario $A$.
\end{remark}
}{
}

In fact, even for $1$-simple PPSs, the $(n,a,b)$-resilience problem is $\Sigma^\P_{2n+1}$-hard. We show this by a reduction from $\Sigma_{2n+1}$-SAT, the language of true quantified Boolean formulas (QBF) with ${2n+1}$ quantifier alternations, where the first quantifier is existential and the underlying propositional formula is in $3$-CNF form. This problem is known to be $\Sigma^\P_{2n+1}$-complete \cite{stockmeyer1976polynomial}. Recall that the truth of a quantified Boolean formula can be analyzed by considering the \emph{QBF evaluation game} for the formula. In this game, two players, \emph{Spoiler} and \emph{Duplicator}, take turns choosing assignments to the formula's quantified variables. Duplicator chooses assignments for existentially-quantified variables with the goal of satisfying the underlying Boolean formula, while Spoiler chooses assignments for universally-quantified variables with the goal of falsifying it. The game concludes once assignments have been chosen for all of the quantified variables. A QBF $\psi$ is true if and only if Duplicator has a winning strategy in this game \cite{papadimitriou07book}.

In our reduction, we encode the positions of this QBF evaluation game into configurations, where a position of the QBF evaluation game for a formula
$$\psi := \exists \overline{v}_1 \forall \overline{v}_2 \exists \overline{v}_3 \hdots \forall \overline{v}_{2n} \exists \overline{v}_{2n+1} \varphi (\overline{v}_1,\overline{v}_2,\overline{v}_3, \hdots, \overline{v}_{2n+1})$$
is a sequence $\mathcal{P} = V_1,\hdots,V_j$ of assignments to the variables in $\overline{v}_1,\hdots,\overline{v}_j$ for some ${j \leq 2n+1}$. If $j$ is even, then we say that the position $\mathcal{P}$ \emph{belongs to Duplicator}; otherwise, we say that it \emph{belongs to Spoiler}. The player who owns a given position makes the next move, choosing an assignment for the variables in the tuple $\overline{v}_{j+1}$. We use system rules to model assignments made by Duplicator, while update rules are used to model assignments made by Spoiler. Intuitively, the goal configurations are those positions of the game which encode assignments satisfying the underlying formula $\varphi$. 

\begin{theorem}
\label{thm:lower-bound}
For all $a \in \mathbb{Z}^+$ and $b \in \mathbb{N}$, there exists a polynomial-time reduction from the $\Sigma_{2n+1}$-$\SAT$ problem to the $(n,a,b)$-resilience problem. Furthermore, the computed instance is always a $1$-simple progressing planning scenario with traces containing only facts of bounded size.
\end{theorem}
\ifthenelse{\boolean{longversion}}{
\begin{proof}
Let $\psi := \exists \overline{v}_1 \forall \overline{v}_2 \exists \overline{v}_3 \hdots \forall \overline{v}_{2n} \exists \overline{v}_{2n+1} \varphi (\overline{v}_1,\overline{v}_2,\overline{v}_3, \hdots, \overline{v}_{2n+1})$ be an instance of $\Sigma_{2n+1}$-SAT, where each $\overline{v}_i = v^i_1, \hdots, v^i_{k_i}$ is a tuple of variables. We assume that $\overline{v}_i$ and $\overline{v}_j$ are disjoint whenever $i \neq j$. Furthermore, we assume that $\varphi = C_1 \land \hdots \land C_m$ is a 3-CNF formula with $m$ clauses and whose $k = \sum_{i \leq 2n+1} k_i$ variables are precisely those in the combined tuple $\overline{v} = \overline{v}_1 \overline{v}_2 \hdots \overline{v}_{2n+1}$. We now describe a $1$-simple PPS $A=(\Rscr,\GS,\CS,\Escr,\Sscr_0)$, which we will show is $(n,1,0)$-resilient if and only if $\psi$ is a true quantified Boolean formula.

\gap
\noindent \textbf{The alphabet} \\
We set $\Sigma_G = \{T(true)\}$ and $\Sigma_C = \emptyset$, and we define
\begin{align*}
\Sigma_S := \{ B, T, F, W, true, false, Junk \} ~&\cup \{ Rnd_i \mid 0 \leq i \leq 2n+1 \} \\
&\cup \{ Unk_i \mid 1 \leq i \leq 2n+1 \} \\
&\cup \{ Val_i \mid 1 \leq i \leq 2n+1 \} \\
&\cup \{ I_{C_j} \mid 1 \leq j \leq m \} \\
&\cup \{ Sat_j \mid 1 \leq j \leq m \},
\end{align*}
where $B$, $T$, and $F$ are unary relation symbols of type $bool$, $W$ is a $0$-ary relation symbol, $true$ and $false$ are constants of type $bool$, each $Val_i$ is a $k_i$-ary relation symbol of type $bool^{k_i}$, and $Junk$, $Rnd_i$, $Unk_i$, $I_{C_j}$ and $Sat_i$ are all $0$-ary relation symbols. 

\gap
\noindent \textbf{Initial / goal / critical configurations} \\
We set $\mathcal{CS} = \emptyset$ and $\mathcal{GS} = \{ \langle W, T(true) \emptyset \rangle, \langle \{T(true)\} \cup \{ Sat_j \mid j \leq m \}, \emptyset \rangle \}$. The initial configuration is
\begin{align*}
\mathcal{S}_0 &:= \{ \Time@0, Rnd_0@0, T(true)@0, F(false)@0 \} \\
&\hspace*{30pt} \cup \{ Unk_i@0 \mid 1 \leq i \leq n+2 \} \\
&\hspace*{30pt} \cup \{ I_{C_j}@0 \mid 1 \leq j \leq m \} \\
&\hspace*{30pt} \cup \{ \underbrace{B(true)@0, \hdots, B(true)@0}_{2k ~\text{copies}} \} \\
&\hspace*{30pt} \cup \{ \underbrace{B(false)@0, \hdots, B(false)@0}_{2k ~\text{copies}} \} \\
&\hspace*{30pt} \cup \{ \underbrace{Junk@0, Junk@0, \hdots, Junk@0}_{2n+m+1 ~\text{copies}} \}.
\end{align*}

\noindent \textbf{System Rules} \\
For each \emph{odd} $i$ with $1 \leq i \leq 2n+1$, we have a system rule $assign^\exists_i$ as follows:
\begin{align*}
\Time@T, ~&B(y_1)@T_1, \hdots, B(y_{k_i})@T_{k_i}, \\
&\red{Rnd_{i-1}@T_{k_i+1}, Unk_i@T_{k_i+2}, Junk@T_{k_i+3}} \\
&\mid \{ T_j = T \mid j \leq k_i+3 \} \\
&\longrightarrow_{assign^\exists_i} \Time@T,
B(y_1)@T_1, \hdots, B(y_{k_i})@T_{k_i}, \\
&\hspace*{50pt} \blue{Rnd_i@T, Val_i(y_1,\hdots,y_{k_i})@T, Junk@(T+1)}.
\end{align*}

\noindent For each \emph{even} $i$ with $2 \leq i \leq 2n+1$, we have a system rule $r^{win}_i$ as follows:
\begin{align*}
\Time@T, ~&\red{Rnd_{i-1}@T_1, Unk_i@T_2, Junk@T_3} \mid \{ T_j = T \mid j \leq 3 \} \\
&\longrightarrow_{r^{win}_i}  \Time@T, \blue{W@T, Junk@T, Junk@(T+1)}.
\end{align*}

\noindent If $v^j_i$ occurs positively in $C_l$ for $l \leq m$, we have a system rule $posElim^{i,j}_l$ as follows:
\begin{align*}
\Time@T, ~&Val_i(y_1,\hdots,y_{j-1},b,y_{j+1},\hdots,y_{k_i})@T_1, T(b)@T_2, Rnd_{2n+1}@T_3, \\
&\red{I_{C_l}@T_4, Junk@T_5} \mid \{ T_j = T \mid j \leq 5 \} \\
&\longrightarrow_{posElim^{i,j}_l}  \Time@T, Val_i(y_1,\hdots,y_{j-1},b,y_{j+1},\hdots,y_{k_i})@T_1, \\
&\hspace*{60pt} T(b)@T_2, Rnd_{2n+1}@T_3, \blue{Sat_l@T, Junk@(T+1)}.
\end{align*}

\noindent If $v^j_i$ occurs negatively in $C_l$ for $l \leq m$, we have a system rule $negElim^{i,j}_l$ as follows:
\begin{align*}
\Time@T, ~&Val_i(y_1,\hdots,y_{j-1},b,y_{j+1},\hdots,y_{k_i})@T_1, F(b)@T_2, Rnd_{2n+1}@T_3, \\
&\red{I_{C_l}@T_4, Junk@T_5} \mid \{ T_j = T \mid j \leq 5 \} \\
&\longrightarrow_{negElim^{i,j}_l}  \Time@T, Val_i(y_1,\hdots,y_{j-1},b,y_{j+1},\hdots,y_{k_i})@T_1, \\
&\hspace*{60pt} F(b)@T_2, Rnd_{2n+1}@T_3, \blue{Sat_l@T, Junk@(T+1)}.
\end{align*}

\noindent \textbf{Update Rules} \\
For each \emph{even} $i$ with $2 \leq i \leq 2n+1$, we have a system update rule $assign^\forall_i$ as follows:
\begin{align*}
\Time@T, ~&B(y_1)@T_1, \hdots, B(y_{k_i})@T_{k_i}, \\
&\red{Rnd_{i-1}@T_{{k_i}+1}, Unk_i@T_{k_i + 2}, Junk@T_{k_i+3}} \mid \ \{ T_j \leq T \mid j = k_i+3 \} \\
&\longrightarrow_{assign^\forall_i}  \Time@T, B(y_1)@T_1, \hdots, B(y_{k_i})@T_{k_i}, \\
&\hspace*{50pt} \blue{Rnd_i@T, Val_i(y_1,\hdots,y_{k_i})@T, Junk@(T+1)}.
\end{align*}

\noindent \textbf{Size of the System} \\
There are $n+1$ \textit{$assign^\exists$} rules, $2n$ \textit{$r^{win}$} rules, $3m$ \textit{Elim} rules (one for each literal of each clause), and $n$ \textit{$assign^\forall$} rules. Hence $\mathcal{R}$ and $\mathcal{E}$ are linear in the size of $\psi$. The size of the critical and goal configuration specifications, $\mathcal{CS}$ and $\mathcal{GS}$, are constant and linear, respectively, in the size of $\psi$. The initial configuration $\mathcal{S}_0$ is also clearly linear in the size of $\psi$. Hence $A$ can be computed in polynomial time in the size of $\psi$.

\gap
\textbf{Correctness of the Reduction} \\
Observe that $A$ is a progressing planning scenario, and since $\CS$ is empty, it is $1$-simple. Furthermore, note that the $Val_i$ facts are the largest facts of the system, but due to the lack of function symbols, are bounded in size. Let $\mathcal{P} = V_1, \hdots, V_j$ be a position of the QBF evaluation game on $\psi$. If $X$ is a strategy for Duplicator (resp. Spoiler), we say that $\mathcal{P}$ is \textit{consistent} with $X$ if, for each even (resp. odd) $i < j$, we have that $V_{i+1}$ is the assignment chosen by Duplicator (resp. Spoiler) from the position $V_1, \hdots, V_i$ when Duplicator (resp. Spoiler) plays according to $X$. For every position $\mathcal{P} = V_1, \hdots, V_j$, we define a corresponding configuration:
\begin{align*}
\mathcal{S}(V_1,\hdots,V_j) = \{ \Time@0, ~&Rnd_j@0, T(true)@0, F(false)@0 \} \\
&\cup \{ Unk_i@0 \mid j < i \leq n+2 \} \\
&\cup \{ Val_i(V_i)@0 \mid 1 \leq i \leq j \} \\
&\cup \{ I_{C_i}@0 \mid 1 \leq i \leq m \} \\
&\cup \{ \underbrace{B(true)@0, \hdots, B(true)@0}_{2k ~\text{copies}} \} \\
&\cup \{ \underbrace{B(false)@0, \hdots, B(false)@0}_{2k ~\text{copies}} \} \\
&\cup \{ \underbrace{Junk@0, Junk@0, \hdots, Junk@0}_{2n+m+1-j ~\text{copies}} \} \\
&\cup \{ \underbrace{Junk@1, Junk@1, \hdots, Junk@1}_{j ~\text{copies}} \},
\end{align*}
where $Val_i(V_i)@0$ denotes the timestamped fact $Val_i(b_1,\hdots,b_{k_i})@0$ where, for each $n \leq k_i$, we have that $b_n = true$ if $V_i(v^i_n) = 1$, and $b_n = 0$ otherwise. The configuration $\mathcal{S}(V_1,\hdots,V_j)$ encodes a position of the QBF evaluation game, in which Duplicator's chosen assignments are the $V_i$ for odd $i$, and Spoiler's chosen assignments are the $V_i$ for even $i$. Note that we can view $\mathcal{S}_0$ as a configuration of this form for the empty sequence of assignments, which corresponds to the starting position of the game. Furthermore, for each odd $j \leq n$, observe that $r^{win}_j$ is applicable to $\mathcal{S}(V_1,\hdots,V_j)$, and so we also define $\mathcal{S}(V_1,\hdots,V_j,W)$ to be the resulting (goal) configuration.

The applicable rules are strictly controlled by the $Rnd_i$ predicates. The only rule applicable to $\mathcal{S}_0$ is $assign^\exists_1$, after which the configuration has the form $\mathcal{S}_1(V_1)$ for some assignment $V_1$ to the variables in $\overline{v}_1$. From configurations of the form $\mathcal{S}_1(V_1)$, both the system rule $r^{win}_2$ and the update rule $assign^\forall_2$ are applicable, leading to configurations of the form $\mathcal{S}_2(V_1,W)$ and $\mathcal{S}_2(V_1,V_2)$ for some assignment $V_2$ to the variables in $\overline{v}_2$, respectively. This is summarized in Figure \ref{fig:base-case-tree}, where blue lines indicate transitions due to system rules, and red lines indicate transitions due to update rules.

\begin{figure}
\centering
\begin{tikzpicture}
[level 1/.style={sibling distance=40mm}, 
 level 2/.style={sibling distance=18mm}]

\node (root) {$S_0$} 
 child {node (s1) {$S_1(V_1)$} 
  child {node (s2w) {$S_2(V_1, W)$} edge from parent[blue]}
  child {node (s2v) {$S_2(V_1, V_2)$} edge from parent[red]}
  edge from parent[blue]};

\node at ($(root) + (-3,0)$) {$Rnd_0$};
\node at ($(s1) + (-3,0)$) {$Rnd_1$};
\node at ($0.5*(s2w) + 0.5*(s2v) + (-3,0)$) {$Rnd_2$};

\end{tikzpicture}
\caption{Possible two-step traces starting from $\mathcal{S}_0$}
\label{fig:base-case-tree}
\end{figure}
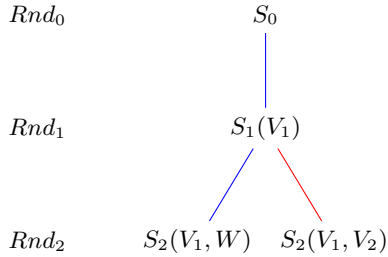

For even $i < 2n$, only $assign^\exists_{i+1}$ is applicable to a configuration of the form $\mathcal{S}_i(V_1,\hdots,V_i)$, after which both $assign^\forall_{i+2}$ and $r^{win}_{i+2}$ are applicable. This is summarized in Figure \ref{fig:ind-step-tree}, where each $V_j$ is an assignment to the appropriate variables.

\begin{figure}
\centering
\begin{tikzpicture}
[level 1/.style={sibling distance=40mm}, 
 level 2/.style={sibling distance=45mm}]

\node (root) {$S_i(V_1,\hdots,V_i)$} 
 child {node (s1) {$S_{i+1}(V_1,\hdots,V_i,V_{i+1})$} 
  child {node (s2w) {$S_{i+2}(V_1,\hdots,V_i,V_{i+1},W)$} edge from parent[blue]}
  child {node (s2v) {$S_{i+2}(V_1,\hdots,V_i,V_{i+1},V_{i+2})$} edge from parent[red]}
  edge from parent[blue]};

\node at ($(root) + (-6,0)$) {$Rnd_i$};
\node at ($(s1) + (-6,0)$) {$Rnd_{i+1}$};
\node at ($0.5*(s2w) + 0.5*(s2v) + (-6,0)$) {$Rnd_{i+2}$};

\end{tikzpicture}
\caption{Possible two-step traces starting from $\mathcal{S}_i(V_1,\hdots,V_i)$}
\label{fig:ind-step-tree}
\end{figure}
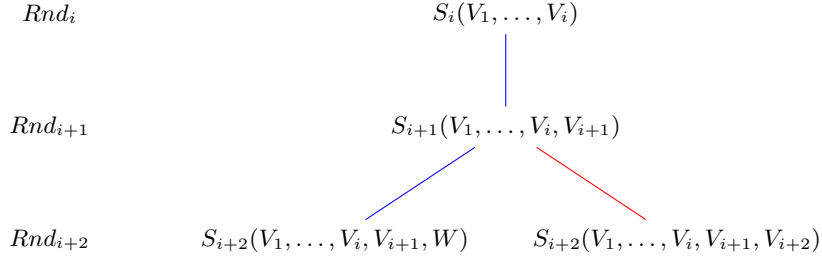

From a configuration of the form $\mathcal{S}_{2n}(V_1,\hdots,V_{2n})$, the only applicable rule is $assign^\exists_{2n+1}$, leading to a configuration of the form $\mathcal{S}_{2n+1}(V_1,\hdots,V_{2n},V_{2n+1})$. After this rule, the only rules that may be applicable are the $posElim^{i,j}_l$ and $negElim^{i,j}_l$ rules. If, for each $l \leq m$, there is some $i \leq 2n+1$ and $j \leq k_i$ such that at least one of $posElim^{i,j}_l$ or $negElim^{i,j}_l$ is applicable to $\mathcal{S}_{2n}(V_1,\hdots,V_{2n},V_{2n+1})$, then there exists a trace with no occurrences of the $Tick$ rule starting from $\mathcal{S}_{2n+1}(V_1,\hdots,V_{2n},V_{2n+1})$ and applying these $m$ rules which leads to a goal configuration containing the $Sat_j$ facts for each $j \leq m$. In other words, if $V = V_1 \cup \hdots \cup V_{2n+1}$ is a satisfying assignment, then there exists a compliant trace from $\mathcal{S}_{2n+1}(V_1,\hdots,V_{2n},V_{2n+1})$ to a goal configuration containing a fact $Sat_j$ for each $j \leq m$. Intuitively, this trace corresponds to verifying that each clause is satisfied by the valuation $V = V_1 \cup \hdots \cup V_{2n+1}$. Note that there are sufficiently many $Junk@0$ facts in the initial configuration to allow the entire game and verification processes to be carried out without any applications of the $\Tick$ rule.

We need to show that $\psi$ is true if and only if there exists an $(n,1,0)$-resilient trace for the PPS $A$. Intuitively, the argument proceeds as follows. Figures \ref{fig:base-case-tree} and \ref{fig:ind-step-tree} encode the entire game tree of the QBF evaluation game, plus some additional leaves at each odd level corresponding to a goal state. The system is Duplicator, and the update rules are Spoiler. If either player has a winning strategy, then they can reach a configuration $\hat{S} = \mathcal{S}_{2n+1}(V_1,\hdots,V_{2n+1})$, where $V = V_1 \cup \hdots \cup V_{2n+1}$ either satisfies or falsifies $\varphi$. If the configuration of this form satisfies $\varphi$, in which case Duplicator had the winning strategy, then the system will reach a goal configuration containing $Sat_j$ for each $j$ with $1 \leq j \leq m$. Otherwise, Spoiler has a winning strategy, and the system will get stuck.

For the forward direction, suppose that $\psi$ is true, and hence Duplicator has a (deterministic) winning strategy $X$ for the QBF evaluation game for $\psi$. We show by induction on $i$ with $0 \leq i \leq n$ that, for any game position $\mathcal{P} = V_1,\hdots,V_{2(n - i)}$, if $\mathcal{P}$ is consistent with $X$, then there exists an $(i,1,0)$-resilient trace starting from $\mathcal{S}_i(V_1,\hdots,V_{2(n - i)})$ with no applications of the $\Tick$ rule. In particular, since the initial position of the game is consistent with $X$, this implies that there exists an $(n,1,0)$-resilient trace starting from $\mathcal{S}_0$, which is what we want to show.

For the base case, $i = 0$, and so we need to show that, if $V_1,\hdots,V_{2n}$ is consistent with $X$, then there exists a compliant trace starting from $\mathcal{S}_i(V_1,\hdots,V_{2n})$ with no applications of the $\Tick$ rule. Let $V_{2n+1}$ be the assignment chosen by Duplicator from the position $V_1,\hdots,V_{2n}$ under strategy $X$. Since $X$ is a winning strategy, it follows immediately that $V = V_1 \cup \hdots \cup V_{2n+1}$ is a satisfying assignment for $\varphi$. Then, by applying $assign^\exists_{2n+1}$, the next configuration is $\mathcal{S}(V_1,\hdots,V_{2n+1})$. Then, by our earlier discussion, since $V$ is a satisfying assignment, there exists a compliant trace from $\mathcal{S}(V_1,\hdots,V_{2n+1})$ to a goal configuration. We have described a compliant trace from $\mathcal{S}_i(V_1,\hdots,V_{2n})$ to a goal configuration with no applications of the $\Tick$ rule, and so we're done.

Now suppose inductively that, for all game positions $\mathcal{P} = V_1,\hdots,V_{2(n - i)}$ consistent with $X$, there exists an $(i,1,0)$-resilient trace from $\mathcal{S}_{2(n-i)}(V_1,\hdots,V_{2(n - i)})$ with no applications of the $\Tick$ rule. Let $\mathcal{P}' = V'_1,\hdots,V'_{2(n-(i+1))}$ be a game position consistent with $X$; we want to show that there exists an $(i+1,1,0)$-resilient trace from $\mathcal{S}_{2(n-(i+1))}(V'_1,\hdots,V'_{2(n-(i+1))})$ with no occurrences of the $\Tick$ rule. Since $2(n-(i+1)) = 2n-2i-2$ is even, position $\mathcal{P}'$ belongs to Duplicator. Let $V'_{2n-2i-1}$ be the assignment chosen by Duplicator from the position $\mathcal{P}'$, and let $\tau$ denote the trace
\begin{align*}
\mathcal{S}_{2n-2i-2}&(V'_1,\hdots,V'_{2(n-(i+1))}) \\
&\longrightarrow_{assign^\exists_{2n-2i-1}} ~~\mathcal{S}_{2n-2i-1}(V'_1,\hdots,V'_{2(n-(i+1))},V'_{2n-2i-1}) \\
&\longrightarrow{r^{win}_{2n-2i}} ~~~~\mathcal{S}_{2n-2i}(V'_1,\hdots,V'_{2(n-(i+1))},V'_{2n-2i-1},W),
\end{align*}
which is clearly compliant and leads to a goal configuration. The only applicable update rule for this trace is $assign^\forall_{2n-2i}$, which is applicable only to
$$\mathcal{S}_{2n-2i-1}(V'_1,\hdots,V'_{2(n-(i+1))},V'_{2n-2i-1}),$$
and results in a configuration of the form
$$\mathcal{S}_{2n-2i-1}(V'_1,\hdots,V'_{2(n-(i+1))},V'_{2n-2i-1},V'_{2n-2i})$$
for some assignment $V'_{2n-2i}$ to the variables in the tuple $\overline{v}_{2n-2i}$. Since the game position
$$V'_1,\hdots,V'_{2n-2i-2},V'_{2n-2i-1},V'_{2n-2i}$$
extends $\mathcal{P}'$ with only one move by Spoiler, it is consistent with $X$. Then by the inductive hypothesis, that there exists an $(i,1,0)$-resilient reaction trace starting from
$$\mathcal{S}_{2n-2i}(V'_1,\hdots,V'_{2(n-(i+1))},V'_{2n-2i-1},V'_{2n-2i})$$
with no occurrences of the $\Tick$ rule. It follows that $\tau$ is a $(i+1,1,0)$-resilient trace starting from $\mathcal{S}_{2(n-(i+1))}(V'_1,\hdots,V'_{2(n-(i+1))})$ with no occurrences of the $\Tick$ rule, which is what we wanted to show. This concludes the inductive argument, from which we can conclude that, if $\psi$ is true, then $A$ is $(n,1,0)$-resilient.

The reverse direction is similar, using the fact that if $\psi$ is false, then Spoiler has a winning strategy $X$. Then, if update rules are applied so that the configuration always corresponds to a game position which is consistent with $X$, the system will reach a configuration $\mathcal{S}_{2n+1}(V_1,\hdots,V_{2n},V_{2n+1})$ where $V = V_1 \cup \hdots \cup V_{2n},V_{2n+1}$ is not a satisfying assignment for $\varphi$. It follows that the system gets stuck (i.e., no $Elim$ rules will be applicable for some clause) and cannot reach a goal configuration. Thus we can conclude that, if $\psi$ is false, then $A$ is \textit{not} $(n,1,0)$-resilient.

We have shown that $A$ is $(n,1,0)$-resilient if and only if $\psi$ is true. Now, note that
\begin{enumerate}
\item the global time in $\mathcal{S}_0$ is $0$;
\item every system and update rule contains a $Rnd_i@T'$ fact in its precondition, and the guard of the rule requires that $T' = T$, where $T$ is the global time; and
\item for any configurations $\Sscr,\Sscr'$ and system or update rule $r$ such that $\Sscr \lra_r \Sscr'$, where $Rnd_i@T_i$ is in $\Sscr$ and $Rnd_j@T_j$ is in $\Sscr'$, we have that $T_i = T_j$.
\end{enumerate}
These observations imply that every $Rnd_i$ fact in any configuration of this system reachable (by system or update rules) from $\mathcal{S}_0$ has timestamp $0$. Furthermore, no system or update rules can be applied if the global time exceeds $0$, since the guard of every rule is violated for configurations where the global time exceeds the timestamp of the unique $Rnd_i$ fact. It follows easily that a trace $\tau$ is $(n,1,0)$-resilient if and only if it is also $(n,a,0)$-resilient for all $a \in \mathbb{Z}^+$. Similarly, since goal configurations can only be reached, and updates only be applied, when the global time is $0$, we conclude that if a trace is $(n,a,b)$-resilient, then it has an $(n,a,0)$-resilient sub-trace. Furthermore, by Remark \ref{remark:parameter-observations}, if a trace is $(n,a,0)$-resilient, then it is also $(n,a,b)$-resilient. Hence
\begin{align*}
\psi ~\text{is true} &\iff A ~\text{is}~ (n,1,0)\text{-resilient} \\
&\iff A ~\text{is}~ (n,a,0)\text{-resilient for all}~ a \in \mathbb{Z}^+ \\
&\iff A ~\text{is}~ (n,a,b)\text{-resilient for all}~ a \in \mathbb{Z}^+ ~\text{and}~ b \in \mathbb{N}
\end{align*}
Therefore, for each $a \in \mathbb{Z}^+$ and $n,b \in \mathbb{N}$, we have that $\Sigma_{2n+1}$-$\SAT$ is polynomial-time reducible to the $(n,a,b)$-resilience problem.\qed
\end{proof}

\ifthenelse{\boolean{longversion}}{
\begin{remark}
In~\cite{kanovich16formats}, we claimed that the Boolean satisfiability problem can be reduced to the $n$-time-bounded-realizability problem. We have since identified an error in the original proof in that paper. However, since $n$-time-bounded realizability is equivalent to $(0,n,0)$-resilience, the above proof of Theorem \ref{thm:lower-bound} constitutes a correction to the proof, since $\Sigma_1$-$\SAT$ is just the standard Boolean satisfiability problem.
\end{remark}
}{
}

\noindent Theorems \ref{thm:upper-bound} and \ref{thm:lower-bound} immediately entail the following.

}{
\input{conf-only/lower-bound-sketch}
}

\begin{corollary}
\label{thm:sigma2n+1-complete}
The $(n,a,b)$-resilience problem for $\eta$-simple PPSs with traces containing only facts of bounded size is $\Sigma^\P_{2n+1}$-complete.
\end{corollary}

\section{Verifying Resilience in Maude}
\label{sec:implementation}
To experiment with resilience, we specified our running example of a travel planning scenario in the Maude rewriting logic language \cite{clavel-etal-07maudebook}. In contrast to the multiset rewriting representation, the Maude specification uses data structures, not facts, to represent system structure and state. The passing of time is modeled using rule duration.  For example, the rule that models taking a flight takes time according to the duration of the flight. These rules combine an instantaneous rule with a time-passing rule. These design decisions, together with relegating as much as possible to equational reasoning, help reduce the search state space.

In the travel system scenarios, the goal is to attend a given set of events.  System updates change flight schedules; goal updates either change event start time or duration, or add an event. Execution traces terminate when the last event is attended or an event is missed.  Thus, for simplicity, we fix $a+b$ to be the end of the last possible event and leave it implicit. In the following, we describe the representation of key elements of the travel system specification: system state, execution rules, and updates. We then explain the algorithm for checking $(n,a,b)$-resilience and report on some simple experiments.

\subsubsection{Representing travel status.}
The two main sorts in the travel system specification are \smalltt{Flight} and \smalltt{Event}. A term \smalltt{fl(cityD,cityA,fn,depT,dur)} represents a flight, where \smalltt{cityD} and \smalltt{cityA} are the departure and arrival cities, \smalltt{fn} is the flight number (a unique identifier), \smalltt{depT} is the departure time, and \smalltt{dur} is the duration. The departure time and duration are represented by hour-minute terms \smalltt{hm(h,m)}. For simplicity, flights are assumed to go at the same time every day, and all times are in GMT. A flight instance (sort \smalltt{FltInst}) represents a flight on a specific date by a term: \smalltt{fi(flt,dtDep,dtArr)} where \smalltt{flt} is a flight, \smalltt{dtDep} and \smalltt{dtArr} are date-time terms representing the date and time of departure and arrival respectively. A date-time term has the form \smalltt{dt(yd,hm)} where \smalltt{yd} is a year-day term \smalltt{yd(y,d)} and \smalltt{hm} is an hour-minute term as above.

An event is represented by a term \smalltt{ev(eid,city,loc,yd,hm,dur,opt)}, where \smalltt{eid} is a unique (string) identifier and the \smalltt{opt} Boolean specifies if the event is optional; the other arguments are as above. \longversion{We form lists of flights (sort \smalltt{FlightL}), sets of flight lists (sort \smalltt{FltLS}), flight instance lists (sort \smalltt{FltInstL}), sets of flight instance lists (sort \smalltt{FltInstLS}) and sets of events (sort \smalltt{Events}) as usual using associative, commutative, and id axioms.} A term of the form \verb|{conf}| (sort \smalltt{Sys}), where \smalltt{conf} (sort \smalltt{Conf}) is a multiset of configuration elements, represents a system execution state. The sorts of configuration elements are \smalltt{TConf}, \smalltt{Log}, and update descriptions. \longversion{A system state is required to have one \smalltt{TConf} element, and there can be at most one log and one update element.} Terms of sort \smalltt{TConf} represent a traveler's state, with one of three forms: \\
\hspace*{3em} \smalltt{tc(dt,city,loc,evs)} -- planning \\
\hspace*{3em} \smalltt{tc(dt,city,loc,evs,ev,fltil)} -- executing \\
\hspace*{3em} \smalltt{tcCrit(dt,city,loc,evs,ev,reason)} -- critical \\
\noindent Here, \smalltt{dt} is a date-time term, the travelers current date and time, and \smalltt{city} tells what city the traveler is currently in. The term \smalltt{loc} gives the location within the city, either the \smalltt{airport} or the city \smalltt{center}. The term \smalltt{evs} is the set of events to be attended, while \smalltt{ev} is the next event to consider. The term \smalltt{fltil} is the flight instance list chosen to get from \smalltt{city} to the location of \smalltt{ev}.  The constructor \smalltt{tcCrit} signals a critical configuration in which a required event has been missed.  

An update description is a term of the form
\smalltt{di(digs)} or \smalltt{di(digs,n)} where \smalltt{digs} is a set of digressions. Each digression describes an update to be applied by an update rule, and \smalltt{n} bounds the number of updates that can be applied. Two kinds of updates are currently implemented: flight/system updates and event/goal updates. The flight updates are: \smalltt{cancel}, which cancels the current flight, \smalltt{delay(hm(h,m))}, which delays the current flight by \smalltt{h} hours, \smalltt{m} minutes, and \smalltt{divert(city0,city1)}, which diverts the current flight from \smalltt{city0} to \smalltt{city1}, where the current flight is the first element of the flight instance list of an executing \smalltt{TConf} term. The event/goal updates are: \smalltt{edEvS(hm(h,m))}, which starts the current event \smalltt{h} hours, \smalltt{m} minutes earlier, \smalltt{edEvD(hm(h,m))}, which extends the current event duration by \smalltt{h} hours, \smalltt{m} minutes, and \smalltt{addEv(eid)}, which adds the event with id \smalltt{evid} to the set of pending events.

Flight updates are only applied to the next flight the traveler is about to take.  Similarly, the changes in event start time or duration are only applied to the next event to attend. This is a simplified setting, but sufficiently illustrates our formalism; more complex variations are possible. Lastly, an element of sort \smalltt{Log} is a list of log items. It is used to record updates, flights taken, and events attended or missed. Among other things, when searching for flights, it is used to know which flights have been canceled.

\subsubsection{Rewrite rules.}
There are five system rules (\smalltt{plan}, \smalltt{noUFlts}, \smalltt{flt}, \smalltt{event}, and \smalltt{replan}) and two update rules (\smalltt{fltDigress} and \smalltt{evDigress}, for flight/system updates and event/goal updates, respectively). The \smalltt{plan} rule picks the event, \smalltt{ev}, with the earliest start time from \smalltt{nevs} and (non-deterministically) selects a list of flight instances, \smalltt{fltil}, from the set of flight instance lists arriving at the event city before the start time. The set of possible flights is stored in a constant \smalltt{FltDB}. \smalltt{log1} is \smalltt{log} with an item recording the rule firing added. The conditional \smalltt{if ...} does the above computing.

\vspace{2pt}
\begin{small}
\smalltt{crl [plan]:  \{tc(dt,city, loc, nevs) log\}}

\hspace{5em} \smalltt{ => \{tc(dt, city, loc, evs0, ev, fltil) log1\}  if ...}  
\end{small}
\\[2pt]
The rule \smalltt{noUFlts} handles cases in which there is no usable flight instance list given the traveler time and location and the time and location of the next event. If \smalltt{ev} is optional then it is dropped (recording this in the log) and the rule \smalltt{plan} is applied to the remaining events; otherwise, the configuration becomes critical and execution terminates.
\omitthis{
\begin{small}
\begin{verbatim}  
 crl[noUFlts]:
  tc(dt,city, at, nevs) log
  =>
  (if getOpt(ev)
  then tc(dt,city,at,evs0) log ; li(getId(ev),false)
  else tcCrit(dt, city, at, evs0, ev, NoUsableFltl") log
  fi)
  if {evs0,ev} := nextEv(nevs)
\end{verbatim}  
\end{small}
}
The rule \smalltt{flt} models taking the next flight, assuming the flight departure time is later than the traveler's current time. This rule updates the traveler's time to the flight arrival time \smalltt{dtArr} and traveler's city to the destination \smalltt{city1}. Then, the flight instance taken is removed from the list. \\
\vspace{2pt}
\begin{small}
\smalltt{
crl [flt]: \{tc(dt, city0, airport, evs,ev, flti; fltil) conf\}}

\smalltt{  =>   \{tc(dtArr, city1, airport, evs, ev, fltil) conf1\}   if  ....}  
\end{small}

\noindent
The \smalltt{event} rule (not shown) models
attending the currently selected event. It can be
applied when the
traveler city is the same as the event city and the
current time is not after the event start. As for
the flight rule, the current time is updated to the
event end time and the traveler returns to the
airport.  If the traveler arrives at the event
city too late, as for the \smalltt{noUFlts} rule,
if the event is optional, the \smalltt{event} rule drops the event and enters a log item, otherwise
it produces a critical \smalltt{TConf}.
\omitthis{
\begin{small}
\begin{verbatim}  
**** attend the event if in time
 crl [event]:
  {tc(dt, city0, airport, evs,ev,fltil) conf}
  =>
  {conf1}
\end{verbatim}  
\end{small}
}
The \smalltt{replan} rule handles the situation in which the traveler city is not the event city and the next flight, if any, does not depart from the traveler city or has been missed. The pending flight instance list is dropped, the event is put back in the event set, and a log item is added to the log.

\omitthis{
\begin{small}
\begin{verbatim}  
 crl [replan]:
  {tc(dt, city0, airport, evs,ev,fltil) conf}
  =>
  {tc(dt, city0, airport, evs ev) conf1 }
  if ...   .
\end{verbatim}  
\end{small}
}

The update rule \smalltt{fltDigress} only applies if the digression counter is greater than zero and the rule decrements the counter.  A flight digression, \smalltt{fdig}, is non-deter\-mi\-nisti\-cally selected from the available digressions (the first argument to \smalltt{di}).

\vspace{2pt}
\begin{small}
\smalltt{
 crl [fltDigress]:
   \{tconf di(fdig digs, s n)  conf\}}
   
  \smalltt{ =>
    \{tconf1 di(digs fdig,n) conf1\} }

\qquad  \smalltt{  if tc(dt, city0, airport, evs,ev, flti ; fltil) := tconf}

\qquad  \smalltt{ $\land$ city0 =/= getCity(ev)}

\qquad  \smalltt{ $\land$ tconf1 conf1 := applyDigression(tconf,conf,fdig)}  
\end{small}

\noindent
The first condition exposes the structure of \smalltt{tconf} to ensure there is a pending flight to update. The auxiliary function \smalltt{applyDigression} specifies the result of the update.  For example, the \smalltt{cancel} update removes \smalltt{flti} from the list and adds a log item recording that this flight instance is cancelled. The case where the update description has the form \smalltt{di(fdig digs)} is similar, except here \smalltt{fdig} is removed when applied, and updating stops when there are no more update elements in the set. Similarly, the rule \smalltt{evDigress} non-deterministically selects an event digression from the configuration's digression set and applies the auxiliary function \smalltt{applyEvDigress} to determine the effect of the update.  It only applies if the update counter is greater than zero, and the \smalltt{TConf} component has a selected next event.

\omitthis{
\begin{small}
\begin{verbatim}  
    crl [evDigress]:
        { tconf di(edig digs, s n)  conf}
        =>
        { tconf1 di(digs edig,n) conf1 }
        if tc(dt, city0, airport, evs,ev, fltil) := tconf
        /\ tconf1 conf1 := applyEvDigress(tconf,conf,edig)
        [print "\nevent digression  " edig "\n ev "  ev] .   
 \end{verbatim}  
 \end{small}
 }
 
A planning scenario is defined by an initial system configuration \smalltt{iSys}, a database of flights \smalltt{FltDB}, a database of events \smalltt{EvDB}, and an update description. A trace \smalltt{iSys -TR-> xSys} is a sequence of applications of rule instances from the travel rules \smalltt{TR}, leading from \smalltt{iSys} to \smalltt{xSys}. It is a \emph{compliant goal trace} if \smalltt{xSys} satisfies the goal condition (\smalltt{goal}) that the traveler configuration has no remaining events, and no required events have been missed. Formally, the traveler component of \smalltt{xSys} must have the form \smalltt{tc(dt,city,loc,mtE)}, since, when a required event is missed, it is rewritten to a term of the form \smalltt{tcCrit(dt,city,loc,evs,ev,comment)}.

\longversion{

\subsubsection{Simple analyses.}
The question of whether or not there exists a compliant trace for a given initial system configuration \texttt{iSys} can be answered by Maude's search command:
\begin{small}
\begin{verbatim}  
    search[1] {iSys nilL} =>!
              {tc:TConf l:Log} such that goal(tc:TConf) .
\end{verbatim}  
\end{small}
\noindent The \texttt{!} after the arrow restricts the search to look for terminal states.  The \texttt{goal} predicate holds only for planning \texttt{TConf} terms in which the events argument is the empty event set. Given an update description \texttt{di(dig)}, we can ask how many reachable terminal states are goal states if this one update is applied at some point using the following search command.
\begin{small}
\begin{verbatim}
    search {iSys nilL di(dig) } =>! 
           {tc:TConf l:Log di(mtD)} such that goal(tc:TConf) .
\end{verbatim}
\end{small}
\noindent Dually, we can ask how many terminal states fail the goal condition.
\begin{small}
\begin{verbatim}  
    search {iSys nilL di(dig) } =>! 
           {tc:TConf l:Log di(mtD)}
                such that not(goal(tc:TConf)) .
\end{verbatim}  
\end{small}
As a warm-up to our experiments with full $(n,a,b)$-resilience in Subsection \ref{subsec:full-nab-implementation}, we begin by studying the following weaker notion. We say that a trace satisfies \textit{simple $n$-resilience} if it reaches a goal state after $n$ updates from a given set were applied at some point in the trace. The following search command tells us the number of terminal goal states reached by simple n-resilient traces

\begin{small}
\begin{verbatim}  
    search {iSys nilL di(allDigs,n) } =>! 
           {tc:TConf l:Log di(allDigs,0)}
                such that goal(tc:TConf) .
\end{verbatim}  
\end{small}

\noindent
where \texttt{allDigs} is the set of all update rules to be considered.

\subsubsection{Simple $n$-resilience experiments.}
We conducted two sets of experiments to check if, by starting early enough, there are traces satisfying simple $n$-resilient to $n$ digressions for different choices of $n$ and different sets of digressions. First, we looked at resilience against a fixed single update to see which updates ``give the most trouble.'' Second, we looked at simple $n$-resilience against flight digressions for various values of $n$.

We considered a scenario with events in 3 cities: \texttt{Ev1} starting at \texttt{DT1}, lasting 4 hours, \texttt{Ev2} starting at \texttt{DT2}, lasting 5 days (120h), and \texttt{Ev3} starting at \texttt{DT3} and lasting 4h. Two start dates were considered: \texttt{DAY0} is chosen to be able to arrive at the first event with a little time to spare, and \texttt{DAY1} is a day earlier. The first set of experiments were instances of the following search pattern.

\begin{center}
\begin{small}
\begin{verbatim}  
    search {tc(dt(yd(23,DAY),hm(16,0)),SF,airport,EVS)
                di(UPDATE) nilL} =>! 
           {tc:TConf di(mtD) l:Log} such that PROP .
\end{verbatim}  
\end{small}
\end{center}

\noindent The results are summarized in Table \ref{table:single-update-experiment}. \texttt{UPDATE} is one of the three flight digressions, PROP is goal(tc:TConf) (indicated by \texttt{g}) or \texttt{not(goal(tc:TConf))} (indicated by \texttt{ng}). In both cases, it is required that the single update was applied using the pattern element  \texttt{di(mtD)}. If \verb|#Ev| is 2, \texttt{EVS} consists of \texttt{Ev1} and \texttt{Ev2}.  If \verb|#Ev| is 3, \texttt{EVS} also includes \texttt{Ev3}. The \texttt{(g,ng)} column gives the number of terminal states that (are,are not) goal states. The \verb|#states| columne gives the total number of states searched, and the \texttt{ms} columns give the total search time in milliseconds.

\begin{table}
\centering
\begin{tabular}{c@{\hspace{6pt}}c@{\hspace{12pt}}c@{\hspace{12pt}}c@{\hspace{12pt}}c@{\hspace{12pt}}c}
\toprule
\textbf{\#Ev} & \textbf{DAY} & \textbf{UPDATE} & \textbf{g,ng} & \textbf{states} & \textbf{ms} \\
\midrule
2 & 247 & none   & 6,0  & 27  & 8.10 \\
  &     & delay  & 12,1 & 65  & 18.16 \\
  &     & cancel & 18,1 & 107 & 24.20 \\
  & 247 & divert & 0,0  & 27  & 9.10 \\
  & 246 & divert & 18,0 & 239 & 28.31 \\
\midrule
3 & 247 & none   & 54,0  & 243  & 35.27 \\
  &     & delay  & 279,0 & 1263 & 100.82 \\
  &     & cancel & 711,0 & 3308 & 210.148 \\
  & 247 & divert & 0,0   & 243  & 27.31 \\
  & 246 & divert & 24,0  & 2183 & 124.102 \\
\bottomrule \\
\end{tabular}
\caption{Summary of single-update resilience experiments}
\label{table:single-update-experiment}
\vspace*{-1.5em}
\end{table}

The state space differs for the different updates because they are applicable under different conditions. The results for \texttt{divert} look different for two reasons. First, it is more disruptive, as can be seen by the increase in goal solutions when using the earlier start date. Second, we have picked a specific pair of cities, and it will be only possible to apply the update if the flight plan includes a flight to the target city.

The second set of experiments were instances of the following search pattern.

\begin{small}
\begin{verbatim}  
    search {tc(dt(yd(23,DAY),hm(16,0)),SF,airport,EVS)
                di(allDi,N) nilL} =>!
           {tc:TConf di(digs:DigressS,0) l:Log} such that PROP .
\end{verbatim}  
\end{small}
\noindent
The results are summarized in Table \ref{table:n-resilience-experiment}. Note that \texttt{N} is the number of update applications, and in both PROP cases it is required that \texttt{N} updates were applied (the counter is 0). The remaining parameters and columns are interpreted as in the previous table.

\begin{table}
\centering
\begin{tabular}{c@{\hspace{6pt}}c@{\hspace{12pt}}c@{\hspace{12pt}}c@{\hspace{12pt}}c@{\hspace{12pt}}c}
\toprule
\textbf{\#Ev} & \textbf{DAY} & \textbf{N} & \textbf{g,ng} & \textbf{states} & \textbf{ms} \\
\midrule
2 & 247 & 1 & 30,2   & 145   & 29.28 \\
  & 246 & 1 & 399,4  & 1790  & 145.115 \\
  & 247 & 2 & 72,2   & 399   & 48.48 \\
  & 246 & 2 & 2208,29 & 10370 & 717.541 \\
  & 247 & 3 & 96,12  & 747   & 77.74 \\
  & 246 & 3 & 8082,198 & 40906 & 2930.2438 \\
\midrule
3 & 247 & 1 & 702,2   & 3025   & 208.158 \\
  & 247 & 2 & 4860,2  & 20973  & 1445.923 \\
  & 246 & 2 & 60912,29 & 263450 & 18019.23278 \\
  & 247 & 3 & 23868,12 & 105597 & 13525.13590 \\
\bottomrule \\
\end{tabular}
\caption{Summary of simple $n$-resilience experiments}
\label{table:n-resilience-experiment}
\vspace*{-1.5em}
\end{table}

\noindent One interesting observation is that the search space is most sensitive to the number of events (intuitively, the ``goal complexity'') and the starting date.
  
}
 
\subsubsection{Checking $(n,a,b)$-resilience.}
\label{subsec:full-nab-implementation}
Checking $(n,a,b)$-resilience in the travel planning system is implemented by the equationally-defined function \smalltt{isAbRes} using Maude's reflection capability and strategy language. As previously mentioned, the upper bound on time is implicitly determined by the times and durations of available events, not treated as a parameter.

The function \smalltt{isAbRes} checks $(n,a,b)$-resilience by first using \smalltt{metaSearch} to find a candidate goal state, then \smalltt{metaSearchPath} gives the corresponding compliant goal trace\footnote{Execution terminates if a critical state is reached, so paths to goal states are always compliant.} The candidate trace is converted to a rewrite strategy (representing the trace's list of rule instances). The function \smalltt{checkAbRes} iterates through the initial prefixes of the strategy, using \smalltt{metaSRewrite} to follow the trace prefix. This implements a check for reaction traces at all possible points of update rule application (cf. Definition \ref{d-trace-ab-recursive}). For each state resulting from executing a prefix, the function \smalltt{checkDigs} is called to apply each one of the available updates, and then we invoke \smalltt{isAbRes} to check for an $(n-1,a-d,b)$-resilient extension trace, where $d$ is the number of time steps up to the end of the prefix. If $n$ is zero, \smalltt{abResCheck} simply finds a compliant goal trace. If no $(n-1,a-d,b)$-resilient extension trace can be found, then the current candidate trace is rejected, and \smalltt{isAbRes} continues searching for the next candidate trace. If the search for candidate traces fails, then the system under consideration is not $(n,a,b)$-resilient. If the check for an $(n-1,a-d,b)$-resilient extension trace succeeds for every update of every prefix execution, then the strategy is returned as a witness for $(n,a,b)$-resilience. We tested $(n,a,b)$-resilience to flight/system updates and event/goal updates with instances of the following command.
\\
\resizebox{\hsize}{!}{\smalltt{ red isAbRes(['TRAVEL-SCENARIO],N,allDi,SYST,patT,tCond,uStrat,0).}
}
\\
\resizebox{\hsize}{!}{\smalltt{
red isAbRes(['TRAVEL-SCENARIO],1,allEv,iSysT,patT,tCond,ueStrat,0).}
}

\noindent
The results are summarized in Table \ref{fig:exp-results-nab}. Note that \smalltt{N} is the number of updates (1, 2, or 3), \smalltt{SYST} is (the meta representation of) an initial state with a starting day that is as late as possible to succeed if nothing goes wrong (247) or one day earlier (246) and 2 or 3 events. \smalltt{patT} and \smalltt{tCond} are the metaSearch pattern and condition arguments and \smalltt{uStrat} is used to construct the update rule strategy. In the summary tables the \smalltt{R?} indicates the result of \smalltt{isAbRes}: \smalltt{Y} for yes (a non-empty strategy is returned) and \smalltt{N} for no. A dash indicates the experiment not done (because the check fails for smaller N). Lastly, \smalltt{allEv} is the set of all implemented event update descriptions.

\begin{figure}[t!]
\centering
\begin{subfigure}{.45\textwidth}
  \centering
\vspace{-5pt}

\begin{tabular}{c@{\hspace{5pt}}c@{\hspace{3pt}}c@{\hspace{3pt}}c@{\hspace{3pt}}c@{\hspace{3pt}}c@{\hspace{3pt}}c@{\hspace{1pt}}}
\toprule
\textbf{N:} & \textbf{1} & & \textbf{2} & & \textbf{3} & \\
\midrule
\textbf{2ev} & \textbf{R?} & \textbf{time} & \textbf{R?} & \textbf{time} & \textbf{R?} & \textbf{time} \\
\midrule
247 & N & 86ms & - & - & - & - \\
246 & Y & 81ms & Y & 147ms & N & 7476ms \\
\midrule
\textbf{3ev} & \textbf{R?} & \textbf{time} & \textbf{R?} & \textbf{time} & \textbf{R?} & \textbf{time} \\
\midrule
247 & N & 1400ms & - & - & - & - \\
246 & Y & 325ms & Y & 685ms & NF & - \\
\bottomrule \\
\end{tabular}

\vspace{-1em}
\caption{flight/system update rules}
\label{fig:gur1}
\end{subfigure}
\qquad \quad 
\begin{subfigure}{.45\textwidth}
  \centering
\vspace{-5pt}

\begin{tabular}{c@{\hspace{5pt}}c@{\hspace{3pt}}c@{\hspace{3pt}}c@{\hspace{3pt}}c@{\hspace{3pt}}c@{\hspace{3pt}}c@{\hspace{1pt}}}
\toprule
\textbf{N:} & \textbf{1} & & \textbf{2} & & \textbf{3} & \\
\midrule
\textbf{2ev} & \textbf{R?} & \textbf{time} & \textbf{R?} & \textbf{time} & \textbf{R?} & \textbf{time} \\
\midrule
247 & Y & 78ms & N & 77ms & - & - \\
246 & Y & 98ms & N & 34800ms & - & - \\
\midrule
\textbf{3ev} & \textbf{R?} & \textbf{time} & \textbf{R?} & \textbf{time} & \textbf{R?} & \textbf{time} \\
\midrule
247 & Y & 143ms & N & 2627ms & - & - \\
246 & Y & 220ms & Y & 633ms & Y & 2634ms \\
\bottomrule \\
\end{tabular}

\vspace{-1em}
\caption{event/goal update rules}
\label{fig:gur2}
\end{subfigure}
\vspace{-3mm}
\caption{Summary of $(n,a,b)$-resilience experiments}
\label{fig:exp-results-nab}
\vspace{-4mm}
\end{figure}

\section{Conclusions and Related Work}
\label{sec:future}
We have shown that, for $\varmax$-simple PPSs with traces containing only facts of bounded size, the $(n,a,b)$-resilience problem is $\Sigma^\P_{2n+1}$-complete. In \cite{Alturki2022resilience}, we showed that the version of this problem without time bounds is $\PSPACE$-complete for balanced systems with traces containing only facts of bounded size. In addition to the formal model and complexity results, we have implemented automated verification of time-bounded resilience using Maude. Resilience has been studied in diverse areas such as civil engineering~\cite{bozza2017urban}, disaster studies~\cite{manyena2006concept}, and environmental science~\cite{folke2006resilience}. Formalizations of resilience are often tailored to specific applications \cite{Sharma-et-al-2014,Banescu-Ochoa-Pretschner-2015,Mouelhi-et-al-2019,Huang-et-al-2020,Madni-Erwin-Sievers-2020} and cannot be easily adapted to different systems. However, while we illustrated time-bounded resilience via an example of a PPS modeling a flight planning scenario, our earlier work has studied similar properties for a diverse range of other critical, time-sensitive systems, from collaborative systems subject to governmental regulations~\cite{kanovich17mscs}, to distributed unmanned aerial vehicles (UAV) performing safety-critical tasks~\cite{kanovich16formats,kanovich2021Guttman}. The strength of our formalism is its flexibility in modeling a wide range of multi-agent systems.

Interest in resilience and related concepts such as robustness \cite{bruneau2003framework}, recoverability \cite{kanovich2021Guttman,nigam-talcott-23tase}, fault tolerance \cite{koren2020fault}, and reliability \cite{bauer2011design}, has grown in recent years. In \cite{zdancewic2001robust,myers2006enforcing}, the authors define a notion of robustness in which a system is robust if the actions taken by an adversary cannot force the system to release more information than it would in the absence of the adversary's actions. In \cite{nigam-talcott-23tase}, the authors study \emph{time-bounded recovery} for \emph{logical scenarios}, which are families of system states represented by patterns whose variables are constrained to describe an operating domain, and where recoverability is parameterized by an ordered set of safety conditions. Intuitively, a $t$-recoverable logical scenario is one which, when operating in normal mode, can recover from a lower-level safety condition to an optimal safety condition within $t$ time steps, without reaching an unsafe condition. Like our formalization, $t$-recoverability concerns recovery from a deviation from normal execution. The primary distinction in \cite{nigam-talcott-23tase} is that deviations are internal to the system, rules update the model state and compute control commands, and enabled rules must fire before time passes, as is common in real-time systems.

Our definition of time-bounded resilience (Definition \ref{d-trace-ab-recursive}) can be seen as a modification of \cite[Definitions 9-10]{Alturki2022resilience}, with time parameters $a$ and $b$ and taking into account both system and goal updates. Here, the recoverability conditions from \cite{Alturki2022resilience} are simplified to the total relation on configurations. Further investigation of recoverability conditions and resilience with respect to update rules that consume or create critical facts is left for future work. Another avenue of future work is to find conditions beyond $\varmax$-simplicity which allow for polynomial-time solvability of the trace compliance problem. We also plan to study time-bounded resilience problems with respect to update rules that consume or create critical facts, and other variations involving, \eg, real-time models and infinite traces. We are interested in the relationship between resilience and other properties of time-sensitive distributed systems~\cite{kanovich2021Guttman}, such realizability, recoverability, reliability, and survivability, as well as in specific applications of resilience, where further implementation results could provide interesting insights.

\subsubsection{Acknowledgments.} We thank Vivek Nigam for many insightful discussions during the early part of this work. Kanovich was partially supported by EPSRC Programme Grant EP/R006865/1: “Interface Reasoning for Interacting Systems (IRIS)”. Scedrov was partially supported by the U. S. Office of Naval Research under award number N00014-20-1-2635 during the early part of this work. Talcott was partially supported by the U. S. Office of Naval Research under award numbers N00014-15-1-2202 and N00014-20-1-2644, and NRL grant N0017317-1-G002.

\bibliographystyle{abbrv}

\begin{thebibliography}{10}

\bibitem{Alturki2022resilience}
M.~A. Alturki, T.~{Ban Kirigin}, M.~Kanovich, V.~Nigam, A.~Scedrov, and C.~Talcott.
\newblock On the formalization and computational complexity of resilience problems for cyber-physical systems.
\newblock In {\em Theoretical Aspects of Computing--ICTAC 2022: 19th International Colloquium, Tbilisi, Georgia, September 27--29, 2022, Proceedings}, pages 96--113. Springer, 2022.

\bibitem{arora2009complexity}
S.~Arora and B.~Barak.
\newblock {\em Complexity theory: A modern approach}.
\newblock Cambridge University Press Cambridge, 2009.

\bibitem{Banescu-Ochoa-Pretschner-2015}
S.~Banescu, M.~Ochoa, and A.~Pretschner.
\newblock A framework for measuring software obfuscation resilience against automated attacks.
\newblock In {\em 2015 IEEE/ACM 1st International Workshop on Software Protection}, pages 45--51, 2015.

\bibitem{bauer2011design}
E.~Bauer.
\newblock {\em Design for reliability: information and computer-based systems}.
\newblock John Wiley \& Sons, 2011.

\bibitem{bennaceur2019modelling}
A.~Bennaceur, C.~Ghezzi, K.~Tei, T.~Kehrer, D.~Weyns, R.~Calinescu, S.~Dustdar, Z.~Hu, S.~Honiden, F.~Ishikawa, Z.~Jin, J.~Kramer, M.~Litoiu, M.~Loreti, G.~Moreno, H.~Müller, L.~Nenzi, B.~Nuseibeh, L.~Pasquale, W.~Reisig, H.~Schmidt, C.~Tsigkanos, and H.~Zhao.
\newblock Modelling and analysing resilient cyber-physical systems.
\newblock In {\em 2019 IEEE/ACM 14th International Symposium on Software Engineering for Adaptive and Self-Managing Systems (SEAMS)}, pages 70--76, 2019.

\bibitem{bloomfield2020towards}
R.~Bloomfield, G.~Fletcher, H.~Khlaaf, P.~Ryan, S.~Kinoshita, Y.~Kinoshit, M.~Takeyama, Y.~Matsubara, P.~Popov, K.~Imai, et~al.
\newblock Towards identifying and closing gaps in assurance of autonomous road vehicles--a collection of technical notes part 1.
\newblock {\em arXiv preprint arXiv:2003.00789}, 2020.

\bibitem{bozza2017urban}
A.~Bozza, D.~Asprone, and F.~Fabbrocino.
\newblock Urban resilience: A civil engineering perspective.
\newblock {\em Sustainability}, 9(1), 2017.

\bibitem{bruneau2003framework}
M.~Bruneau, S.~E. Chang, R.~T. Eguchi, G.~C. Lee, T.~D. O'Rourke, A.~M. Reinhorn, M.~Shinozuka, K.~Tierney, W.~A. Wallace, and D.~Von~Winterfeldt.
\newblock A framework to quantitatively assess and enhance the seismic resilience of communities.
\newblock {\em Earthquake spectra}, 19(4):733--752, 2003.

\bibitem{caminiti2017resilient}
S.~Caminiti, I.~Finocchi, E.~G. Fusco, and F.~Silvestri.
\newblock Resilient dynamic programming.
\newblock {\em Algorithmica}, 77(2):389–425, Feb 2017.

\bibitem{clavel-etal-07maudebook}
M.~Clavel, F.~Dur\'an, S.~Eker, P.~Lincoln, N.~Mart\'i-Oliet, J.~Meseguer, and C.~Talcott.
\newblock {\em All About Maude: A High-Performance Logical Framework}, volume 4350 of {\em LNCS}.
\newblock Springer, 2007.

\bibitem{cunningham2014resilient}
D.~Cunningham, D.~Grove, B.~Herta, A.~Iyengar, K.~Kawachiya, H.~Murata, V.~Saraswat, M.~Takeuchi, and O.~Tardieu.
\newblock Resilient x10: Efficient failure-aware programming.
\newblock {\em SIGPLAN Not.}, 49(8):67–80, feb 2014.

\bibitem{durgin04jcs}
N.~A. Durgin, P.~Lincoln, J.~C. Mitchell, and A.~Scedrov.
\newblock Multiset rewriting and the complexity of bounded security protocols.
\newblock {\em Journal of Computer Security}, 12(2):247--311, 2004.

\bibitem{eigner2021towards}
O.~Eigner, S.~Eresheim, P.~Kieseberg, L.~D. Klausner, M.~Pirker, T.~Priebe, S.~Tjoa, F.~Marulli, and F.~Mercaldo.
\newblock Towards resilient artificial intelligence: Survey and research issues.
\newblock In {\em 2021 IEEE International Conference on Cyber Security and Resilience (CSR)}, pages 536--542, 2021.

\bibitem{ferarro2010experimental}
U.~Ferraro-Petrillo, I.~Finocchi, and G.~F. Italiano.
\newblock Experimental study of resilient algorithms and data structures.
\newblock In P.~Festa, editor, {\em Experimental Algorithms}, pages 1--12, Berlin, Heidelberg, 2010. Springer Berlin Heidelberg.

\bibitem{folke2006resilience}
C.~Folke.
\newblock Resilience: The emergence of a perspective for social–ecological systems analyses.
\newblock {\em Global Environmental Change}, 16(3):253--267, 2006.
\newblock Resilience, Vulnerability, and Adaptation: A Cross-Cutting Theme of the International Human Dimensions Programme on Global Environmental Change.

\bibitem{garey1979computers}
M.~R. Garey and D.~S. Johnson.
\newblock {\em Computers and Intractability: A Guide to the Theory of NP-Completeness (Series of Books in the Mathematical Sciences)}.
\newblock W. H. Freeman, first edition edition, 1979.

\bibitem{goel2024adversarial}
S.~Goel, S.~Hanneke, S.~Moran, and A.~Shetty.
\newblock Adversarial resilience in sequential prediction via abstention.
\newblock {\em Advances in Neural Information Processing Systems}, 36, 2024.

\bibitem{hirshfeld2004logics}
Y.~Hirshfeld and A.~Rabinovich.
\newblock Logics for real time: Decidability and complexity.
\newblock {\em Fundamenta Informaticae}, 62(1):1--28, 2004.

\bibitem{Huang-et-al-2020}
W.~Huang, Y.~Zhou, Y.~Sun, A.~Banks, J.~Meng, J.~Sharp, S.~Maskell, and X.~Huang.
\newblock Formal verification of robustness and resilience of learning-enabled state estimation systems for robotics, 2020.

\bibitem{hukerikar2012programming}
S.~Hukerikar, P.~C. Diniz, and R.~F. Lucas.
\newblock A programming model for resilience in extreme scale computing.
\newblock In {\em IEEE/IFIP International Conference on Dependable Systems and Networks Workshops (DSN 2012)}, pages 1--6, 2012.

\bibitem{kanovich16formats}
M.~Kanovich, T.~{Ban Kirigin}, V.~Nigam, A.~Scedrov, and C.~Talcott.
\newblock Timed multiset rewriting and the verification of time-sensitive distributed systems.
\newblock In {\em 14th International Conference on Formal Modeling and Analysis of Timed Systems (FORMATS)}, 2016.

\bibitem{kanovich2021Guttman}
M.~Kanovich, T.~{Ban Kirigin}, V.~Nigam, A.~Scedrov, and C.~Talcott.
\newblock On the complexity of verification of time-sensitive distributed systems.
\newblock In D.~Dougherty, J.~Meseguer, S.~A. M{\"o}dersheim, and P.~Rowe, editors, {\em Protocols, Strands, and Logic}, volume 13066 of {\em Springer LNCS}, pages 251--275. Springer International Publishing, 2021.

\bibitem{kanovich17jcs}
M.~Kanovich, T.~{Ban Kirigin}, V.~Nigam, A.~Scedrov, and C.~L. Talcott.
\newblock Time, computational complexity, and probability in the analysis of distance-bounding protocols.
\newblock {\em Journal of Computer Security}, 25(6):585--630, 2017.

\bibitem{kanovich17mscs}
M.~Kanovich, T.~{Ban Kirigin}, V.~Nigam, A.~Scedrov, C.~L. Talcott, and R.~Perovic.
\newblock A rewriting framework and logic for activities subject to regulations.
\newblock {\em Mathematical Structures in Computer Science}, 27(3):332--375, 2017.

\bibitem{koren2020fault}
I.~Koren and C.~M. Krishna.
\newblock {\em Fault-tolerant systems}.
\newblock Morgan Kaufmann, 2020.

\bibitem{koutsoukos2018sure}
X.~Koutsoukos, G.~Karsai, A.~Laszka, H.~Neema, B.~Potteiger, P.~Volgyesi, Y.~Vorobeychik, and J.~Sztipanovits.
\newblock Sure: A modeling and simulation integration platform for evaluation of secure and resilient cyber–physical systems.
\newblock {\em Proceedings of the IEEE}, 106(1):93--112, 2018.

\bibitem{Madni-Erwin-Sievers-2020}
A.~M. Madni, D.~Erwin, and M.~Sievers.
\newblock Constructing models for systems resilience: Challenges, concepts, and formal methods.
\newblock {\em Systems}, 8(1), 2020.

\bibitem{madni2009towards}
A.~M. Madni and S.~Jackson.
\newblock Towards a conceptual framework for resilience engineering.
\newblock {\em IEEE Systems Journal}, 3(2):181--191, 2009.

\bibitem{manyena2006concept}
S.~B. Manyena.
\newblock The concept of resilience revisited.
\newblock {\em Disasters}, 30(4):434--450, 2006.

\bibitem{Mouelhi-et-al-2019}
S.~Mouelhi, M.-E. Laarouchi, D.~Cancila, and H.~Chaouchi.
\newblock Predictive formal analysis of resilience in cyber-physical systems.
\newblock {\em IEEE Access}, 7:33741--33758, 2019.

\bibitem{myers2006enforcing}
A.~C. Myers, A.~Sabelfeld, and S.~Zdancewic.
\newblock Enforcing robust declassification and qualified robustness.
\newblock {\em Journal of Computer Security}, 14(2):157--196, 2006.

\bibitem{neches2013towards}
R.~Neches and A.~M. Madni.
\newblock Towards affordably adaptable and effective systems.
\newblock {\em Systems Engineering}, 16(2):224--234, 2013.

\bibitem{nigam-talcott-23tase}
V.~Nigam and C.~L. Talcott.
\newblock Automating recoverability proofs for cyber-physical systems with runtime assurance architectures.
\newblock In C.~David and M.~Sun, editors, {\em 17th International Symposium on Theoretical Aspects of Software Engineering}, volume 13931 of {\em Lecture Notes in Computer Science}, pages 1--19. Springer, 2023.

\bibitem{olowononi2021resilient}
F.~O. Olowononi, D.~B. Rawat, and C.~Liu.
\newblock Resilient machine learning for networked cyber physical systems: A survey for machine learning security to securing machine learning for cps.
\newblock {\em IEEE Communications Surveys \& Tutorials}, 23(1):524--552, 2021.

\bibitem{papadimitriou07book}
C.~H. Papadimitriou.
\newblock {\em Computational complexity}.
\newblock Academic Internet Publ., 2007.

\bibitem{prasad2022towards}
A.~Prasad.
\newblock {\em {Towards Robust and Resilient Machine Learning}}.
\newblock PhD thesis, Carnegie Mellon University, Apr 2022.

\bibitem{Sharma-et-al-2014}
V.~C. Sharma, A.~Haran, Z.~Rakamaric, and G.~Gopalakrishnan.
\newblock Towards formal approaches to system resilience.
\newblock In {\em 2013 IEEE 19th Pacific Rim International Symposium on Dependable Computing}, pages 41--50, 2013.

\bibitem{stockmeyer1976polynomial}
L.~J. Stockmeyer.
\newblock The polynomial-time hierarchy.
\newblock {\em Theoretical Computer Science}, 3(1):1--22, 1976.

\bibitem{vardi2020efficiency}
M.~Vardi.
\newblock Efficiency vs. resilience: What covid-19 teaches computing.
\newblock {\em Communications of the ACM}, 63(5):9--9, 2020.

\bibitem{zdancewic2001robust}
S.~Zdancewic and A.~C. Myers.
\newblock Robust declassification.
\newblock In {\em Proceedings of the 14th IEEE Workshop on Computer Security Foundations}, CSFW '01, page~5, USA, 2001. IEEE Computer Society.

\end{thebibliography}

\end{document}